\documentclass[12pt, draftclsnofoot, onecolumn]{IEEEtran}

\pdfoutput=1

\usepackage{cite}
\usepackage{amsmath,amssymb,amsfonts}
\usepackage{algorithmic}
\usepackage[pdftex]{graphicx}
\usepackage{epstopdf}
\usepackage[english]{babel}
\usepackage{amsthm}
\usepackage{textcomp}

\def\BibTeX{{\rm B\kern-.05em{\sc i\kern-.025em b}\kern-.08em
    T\kern-.1667em\lower.7ex\hbox{E}\kern-.125emX}}

\ifCLASSOPTIONcompsoc
    \usepackage[caption=false, font=normalsize, labelfont=sf, textfont=sf]{subfig}
\else
\usepackage[caption=false, font=footnotesize]{subfig}
\fi

\newtheorem{theorem}{Theorem}
\newtheorem{lemma}[theorem]{Lemma}

\newcommand{\Nt}{N_{\sf{T}}}
\newcommand{\nt}{n_{\sf{T}}}

\newcommand{\am}{\alpha_{\sf{M}}}

\newcommand{\fdma}{\mathbf{f}_{\sf{DMA}}}

\newcommand{\x}{\sf{x}}
\newcommand{\y}{\sf{y}}

\newcommand{\mQ}{\mathcal{Q}}

\newcommand{\mbl}{m_{d,\ell}}
\newcommand{\fps}{\mathbf{f}_{\sf{PS}}}
\newcommand{\heff}{\mathbf{h}_{\sf{eff}}}

\newcommand{\qbl}{q_{d,\ell}}
\newcommand{\wbl}{w_{d,\ell}}

\newcommand{\wtbl}{\Tilde{w}_{d,\ell}}

\newcommand{\hbl}{h_{d,\ell}}
\newcommand{\sfj}{{\sf{j}}}

\DeclareMathOperator*{\argmax}{argmax}
\DeclareMathOperator*{\argmin}{argmin}
    
\begin{document}

\title{Hierarchical Codebook Design with Dynamic Metasurface Antennas for Energy-Efficient Arrays}

\author{Joseph Carlson, Miguel R. Castellanos, \IEEEmembership{Member, IEEE}, and \\ Robert W. Heath, Jr., \IEEEmembership{Fellow, IEEE}

\thanks{J. Carlson, M. R. Castellanos, and  R. W. Heath are with the Department of Electrical and Computer Engineering, North Carolina State University, Raleigh, NC 27606 USA (e-mail: jmcarls3@ncsu.edu; mrcastel@ncsu.edu; rwheathjr@ncsu.edu). This material is based upon work supported by the National Science Foundation under grant nos. NSF-ECCS-2153698, NSF-CCF-2225555, NSF-CNS-2147955 and is supported in part by funds from federal agencies and industry partners as specified in the Resilient \& Intelligent NextG Systems (RINGS) program. This work is also supported in part by a gift from Motorola Mobility, a Lenovo company.}
}

\maketitle

\begin{abstract}
Dynamic metasurface antennas (DMA) provide a solution to form compact, cost-effective, energy-efficient multiple-input-multiple output (MIMO) arrays. In this paper, we implement a practical hierarchical codebook with a realistic DMA design through electromagnetic simulations. We leverage existing DMA models to derive a novel method for enhancing the beamforming gain. We find that the proposed method provides better coverage and spectral efficiency results than prior methods. We also present and verify a new technique for creating wide beamwidths through the DMA and hierarchical codebook. Additionally, we use a detailed transmitter architecture model to determine the power consumption savings of the DMA compared to a typical phased array. The DMA largely outperforms a passive phased array in terms of spectral and energy efficiency due to high component loss from a high-resolution passive phase shifter. While the DMA provides lower spectral efficiency results than the active phased array, the DMA achieves a higher energy efficiency because of the significant power consumption for the active phase shifters. Therefore, we find that DMAs in a realistic wireless environment provide sufficient coverage and spectral efficiency compared to typical phased arrays while maintaining a substantially lower power consumption.

\end{abstract}


\section{Introduction}

DMAs are a type of reconfigurable antenna that form densely-packed energy-efficient arrays \cite{SmithEtAlAnalysisWaveguideFedMetasurfaceAntenna2017}. The DMA consists of slots in a waveguide that radiate power from an initial waveguide excitation, forming a leaky-wave slot antenna. DMAs leverage metasurface technology in their slotted element design to enable tight spacing, and reconfigurable components to perform low-power beamforming. Unlike a traditional antenna that has a fixed radiation pattern, a DMA has an adaptive radiation pattern tuned through the reconfigurable component. This differs from a conventional radio frequency (RF) system, where antennas are connected to separate beamforming hardware like analog phase-shifters \cite{8503216}. Although  DMAs have desirable features from a power consumption and adaptive beamforming perspective \cite{BoyarskyEtAlElectronicallySteeredMetasurfaceAntenna2021,SmithEtAlAnalysisWaveguideFedMetasurfaceAntenna2017}, further research is needed to analyze the benefits and drawbacks of DMAs in a MIMO wireless system. 

\subsection{Prior work} 

Large antenna arrays have been proposed in massive MIMO and millimeter wave (mmWave) systems as a method to increase spectral efficiency \cite{MarzettaNoncooperativeCellularWirelessUnlimited2010,LarssonEtAlMassiveMIMONextGeneration2014,GaoEtAlEnergyEfficientHybridAnalogDigital2016}. Scaling the number of antennas increases the spatial multiplexing capabilities of the system to achieve higher data rates and also allows for radiated energy to be focused very precisely towards desired users \cite{LarssonEtAlMassiveMIMONextGeneration2014}. The increased flexibility, control and directivity of the radiation pattern is vital to compensate for path-loss in mmWave systems \cite{HeathEtAlOverviewSignalProcessingTechniques2016,XiaoEtAlSurveyMillimeterWaveBeamformingEnabled2022} and generally provides more energy-efficient communication as less radiated energy is wasted in undesired directions \cite{RusekEtAlScalingMIMOOpportunitiesChallenges2013}. Despite these advantages, a main challenge in realizing large antenna arrays is the associated increase in power consumption, hardware footprint and overall cost \cite{BjornsonEtAlMassiveMIMOSystemsNonIdeal2014,YangEtAlHardwareConstrainedMillimeterWaveSystems5G2019}. 

Hybrid and analog precoding can decrease the power consumption for an array by reducing the number of radio frequency (RF) chains from conventional digital precoding \cite{ParkEtAlDynamicSubarraysHybridPrecoding2017,AhmedEtAlSurveyHybridBeamformingTechniques2018,AlkhateebEtAlMIMOPrecodingCombiningSolutions2014}. Hybrid precoding has been analyzed for multiple different analog and digital precoding architectures to reduce power consumption, such as partially- \cite{HanEtAlLargescaleAntennaSystemsHybrid2015} and fully-connected phase shifter systems \cite{RialEtAlHybridMIMOArchitecturesMillimeter2016} and subarray architectures \cite{ParkEtAlDynamicSubarraysHybridPrecoding2017}. However, scaling hybrid and analog precoding for large antenna arrays will still consume significant amounts of power through active phase shifters. Instead of concentrating on the transmitter hardware, DMAs have been proposed as a solution to decrease system power consumption by beamforming with the adaptive antenna radiation pattern \cite{SleasmanEtAlWaveguideFedTunableMetamaterialElement2016}. Since DMAs enable beamforming with low-power low-profile reconfigurable components, transmitter hardware components and power consumption can be significantly reduced for large antenna arrays. 

DMAs are a type of reconfigurable antenna, which generally operate by tuning the resonant frequency of individual elements to elicit a desired radiated amplitude and phase for beamforming \cite{HauptLanaganReconfigurableAntennas2013,CostantineEtAlReconfigurableAntennasDesignApplications2015}. A common reconfigurable antenna design is to modify conventional antennas, like dipole or patch antennas, with a reconfigurable component \cite{BehdadSarabandiDualbandReconfigurableAntennaVery2006,PanagamuwaEtAlFrequencyBeamReconfigurableAntenna2006,PiazzaEtAlDesignEvaluationReconfigurableAntenna2008a}. More complex designs involve physically change the antenna geometry via mechanical devices \cite{MazloumanEtAlPatternReconfigurableSquareRing2011} and using materials like ferrites \cite{DixitPourushRadiationCharacteristicsSwitchableFerrite2000} or liquid crystal \cite{LiuLangleyLiquidCrystalTunableMicrostrip2008} to alter the radiated response of the antenna. We focus on the DMA implementation of reconfigurable antennas over other designs as DMAs typically provide higher radiation efficiency and leverage metasurface technology to allow for dense arrays due to the tightly spaced elements \cite{shlezinger2021dynamic}. The dense energy-efficient arrays made possible through DMAs makes them a great candidate to enable the large antenna arrays required for massive MIMO and mmWave systems.

A key challenge in enabling beamforming using a DMA is the element weight constraint. DMA element weights have a limited phase range of $[-\pi,0]$, whereas typical phase shifters provide the full $[0,2\pi]$ phase coverage \cite{SmithEtAlAnalysisWaveguideFedMetasurfaceAntenna2017,SleasmanEtAlWaveguideFedTunableMetamaterialElement2016}. The DMA element weight phases are also coupled to a unique amplitude curve that varies from 0 to 1 with the phase. Prior work has studied DMAs to overcome the DMA weight limitation and enable single-beam beamforming. For desired weights given by unit-amplitude complex-valued weights, the DMA beamforming method proposed in \cite{BowenEtAlOptimizingPolarizabilityDistributionsMetasurface2022,SmithEtAlAnalysisWaveguideFedMetasurfaceAntenna2017} maps the desired weights to the DMA weights with two formulations for minimizing the Euclidean distance of the DMA weights and minimizing the phase difference. Results with a simulated DMA show that both mapping techniques allow for the DMA beam pattern to be steered in a desired direction with a high degree of accuracy in spite of the limited DMA phase range. A physical DMA prototype has also been developed in \cite{BoyarskyEtAlElectronicallySteeredMetasurfaceAntenna2021} to use the mapping techniques experimentally and verify their effectiveness for beamsteering. Additional studies have focused on DMA array processing to surpress grating lobes by applying phase shifts to each waveguide \cite{BoyarskyEtAlGratingLobeSuppressionMetasurface2020} and to minimize side lobes through an optimization algorithm of the reconfigurable controls \cite{JohnsonEtAlSidelobeCancelingReconfigurableHolographic2015}. The foundational work on DMAs in \cite{SmithEtAlAnalysisWaveguideFedMetasurfaceAntenna2017,SleasmanEtAlWaveguideFedTunableMetamaterialElement2016,BowenEtAlOptimizingPolarizabilityDistributionsMetasurface2022,BoyarskyEtAlElectronicallySteeredMetasurfaceAntenna2021} provides methods and results to overcome the beamforming weight limitations imposed by DMA elements, showing the potential for DMAs to create low-power, energy-efficient, reconfigurable arrays.

Recent studies have developed precoding and signal processing algorithms for MIMO with DMAs as well \cite{ShlezingerEtAlDynamicMetasurfaceAntennasUplink2019,WangEtAlDynamicMetasurfaceAntennasMIMOOFDM2021,YouEtAlEnergyEfficiencyMaximizationMassive2022,HuangEtAlJointMicrostripSelectionBeamforming2022}. The precoding studies largely model DMAs as an antenna array with constrained beamforming weights introduced by the DMA characteristics. The waveguide architecture of DMAs is modeled by integrating the waveguide electromagnetic field propagation into the wireless channel. Through this DMA model, prior work has developed ways to optimize the DMA weights for uplink and multi-user MIMO-OFDM scenarios \cite{ShlezingerEtAlDynamicMetasurfaceAntennasUplink2019,WangEtAlDynamicMetasurfaceAntennasMIMOOFDM2021}. An algorithm has been developed to tune the DMA weights in a DMA-based precoding architecture to maximize energy efficiency in a multi-user MIMO system \cite{YouEtAlEnergyEfficiencyMaximizationMassive2022}. A hybrid DMA architecture has been proposed to combine digital precoding with DMA-based precoding \cite{HuangEtAlJointMicrostripSelectionBeamforming2022}. The DMA hybrid precoder was then designed to maximize spectral efficiency via a joint optimization of the digital and DMA weights. While useful for developing signal processing algorithms for DMAs, the aforementioned DMA model places assumptions on the DMA to ignore mutual coupling, element perturbation of the waveguide feed, and other physical non-idealities of the DMA \cite{ShlezingerEtAlDynamicMetasurfaceAntennasUplink2019,WangEtAlDynamicMetasurfaceAntennasMIMOOFDM2021,YouEtAlEnergyEfficiencyMaximizationMassive2022,HuangEtAlJointMicrostripSelectionBeamforming2022}. Additional research is needed to incorporate realistic DMA constraints into the precoder design to obtain algorithms that work in practice.

Energy efficiency is a key metric of performance for MIMO systems due to the large amount of integrated circuitry \cite{FengEtAlSurveyEnergyefficientWirelessCommunications2013,PrasadEtAlEnergyEfficiencyMassiveMIMOBased2017}. Despite the power consumption savings made possible by DMAs, there is still limited work comparing a DMA transmitter architecture to that of analog, digital or hybrid precoding architectures. The power consumption and component loss of the RF front end is heavily dependent on the transmitter architecture and MIMO application, making the development of an accurate hardware model a difficult task. Prior work on MIMO energy efficiency has implemented a realistic power consumption model to derive an expression for the optimal number of transmit antennas per base station to maximize energy efficiency \cite{HaEtAlEnergyEfficiencyAnalysisCircuit2013,DongEtAlEnergyEfficiencyAnalysisCircuit2016}, and show that energy efficiency is maximized for hundreds of antenna per base station in a multi-user environment \cite{BjornsonEtAlOptimalDesignEnergyEfficientMultiUser2015}. Additional research places emphasis on the hardware power consumption for different precoding architectures. A tradeoff between spectral and energy efficiency for active and passive phase shifters was analyzed in \cite{RibeiroEtAlEnergyEfficiencyMmWaveMassive2018,TsaiNatarajan60GHzPassiveActiveRFpath2009}, and switches have been shown to provide channel estimation performance equal to or better than phase shifters \cite{RialEtAlHybridMIMOArchitecturesMillimeter2016}. The insights from \cite{HaEtAlEnergyEfficiencyAnalysisCircuit2013,BjornsonEtAlOptimalDesignEnergyEfficientMultiUser2015} stress the importance of large arrays for energy-efficient MIMO communications, and the hardware models in \cite{RibeiroEtAlEnergyEfficiencyMmWaveMassive2018,RialEtAlHybridMIMOArchitecturesMillimeter2016,TsaiNatarajan60GHzPassiveActiveRFpath2009} provide a valuable baseline for comparing the power consumption of different MIMO precoding architectures. We will build upon these studies to analyze a DMA transmitter architecture compared to an analog precoding architecture.

\subsection{Contributions}

In this paper, we integrate a codebook design with DMAs to analyze the resulting spectral and energy efficiency in a realistic wireless simulation environment. First, we propose a novel method for enhancing the beamforming gain of DMAs through an optimal phase rotation of the desired weights prior to mapping. We then apply this proposed mapping technique to a hierarchical codebook. Despite the widespread use of codebooks in current wireless standards \cite{ZhouOhashiEfficientCodebookbasedMIMOBeamforming2012,UwaechiaMahyuddinComprehensiveSurveyMillimeterWave2020}, there are no studies that investigate codebook design for DMAs. Moreover, we use the hierarchical codebook to present new methods for generating wide beamwidth beam patterns with the DMA. While much of the DMA signal processing literature uses DMA models that ignore aspects like mutual coupling and element perturbation of the waveguide, we simulate a DMA to obtain realistic beam patterns in the electromagnetic simulation software HFSS. This allows us to account for the constraints and effects imposed by the physical DMA design that are not included in typical DMA models. Furthermore, we use the HFSS measurements and channel generation software QuaDRiGa to model a DMA-based MIMO system for spectral efficiency analysis. We apply QuaDRiGa with DMAs to account for channel effects like scattering, polarization, and path loss. Many current studies have implemented a simple multi-path channel model that does not account for these effects \cite{HuangEtAlJointMicrostripSelectionBeamforming2022,ShlezingerEtAlDynamicMetasurfaceAntennasUplink2019,YouEtAlEnergyEfficiencyMaximizationMassive2022,WangEtAlDynamicMetasurfaceAntennasMIMOOFDM2021}. Lastly, we integrate a detailed transmitter architecture to determine the power consumption savings of DMAs compared to a typical phased array. Prior work has analyzed energy efficiency with a DMA model that assumes perfect channel state information and unconstrained precoding \cite{YouEtAlEnergyEfficiencyMaximizationMassive2022}. We incorporate simulated DMA results from a realistic DMA design applying the hierarchical codebook design and limited feedback precoding. We use these calculations to establish the tradeoff in DMA performance with power consumption through the resulting spectral and energy efficiency.

This paper is organized as follows: In Section \ref{sec: system model}, we discuss the operation of the DMA and how to incorporate its radiation characteristics into the signal model. We then establish a MISO-OFDM signal model for use with the DMA, integrating the DMA characteristics into the beamformer and channel model. Next, in Section \ref{sec: DMA mapping}, we describe the DMA mapping techniques implemented in \cite{BowenEtAlOptimizingPolarizabilityDistributionsMetasurface2022} and propose a novel method to enhance the DMA beamforming gain via a phase rotation of the desired weights prior to mapping. In Section \ref{sec: DMA design and hardware}, we discuss the physical DMA design and transmitter architecture used to create DMA beam patterns and analyze power consumption. This section also details the hierarchical codebook design. In Section \ref{sec: results}, we present results for the simulated DMA in HFSS with the hierarchical codebook design and power consumption model. We analyze the ability of the DMA to produce wide beamwidths, the coverage of the DMA hierarchical codebook through multiple mapping techniques, and the spectral and energy efficiency results for the DMA compared to a realistic phased array. Lastly, we summarize the key findings and results for this work in Section \ref{sec: conclusion}.

{\it{Notation}}: We denote a bold, capital letter $ \mathbf{A} $ as a matrix, a bold, lowercase letter $ \mathbf{a} $ as a vector, and a script letter $ \mathcal{A} $ as a set. Let $\mathbf{A}^*$ represent the matrix conjugate transpose, $\mathbf{A}^c$ represent the matrix conjugate, and $\mathbf{A}^T$ represent the matrix transpose. We denote $\sf{vec}(\mathbf{A})$ as the vectorization of the matrix $\mathbf{A}$. For a complex number $a$, we use $\operatorname{Re}\{a\}$ to indicate its real part. We define the operator $\odot$ as the Hadamard product.

\section{System model}\label{sec: system model}

We consider a single-user MISO-OFDM system with DMAs at the transmitter. We take an individual DMA to consist of a planar array with $\Nt$ total radiating elements. The DMA comprises $D$ waveguides and $L$ radiating elements per waveguide such that $\Nt=DL$, as shown in Fig. \ref{fig: dma array}. The DMA is assumed to have reconfigurable elements that allow for changes to their radiated amplitude and phase. For a total of $ K $ subcarriers, we define the $k$th subcarrier wireless channel $ \mathbf{h}[k]  $ between the transmit and receive antennas to be an $ \Nt \times 1 $ matrix \cite{ParkEtAlDynamicSubarraysHybridPrecoding2017}.

% was .95
\begin{figure}[h!]
    \centering
    \includegraphics[width=.6\linewidth]{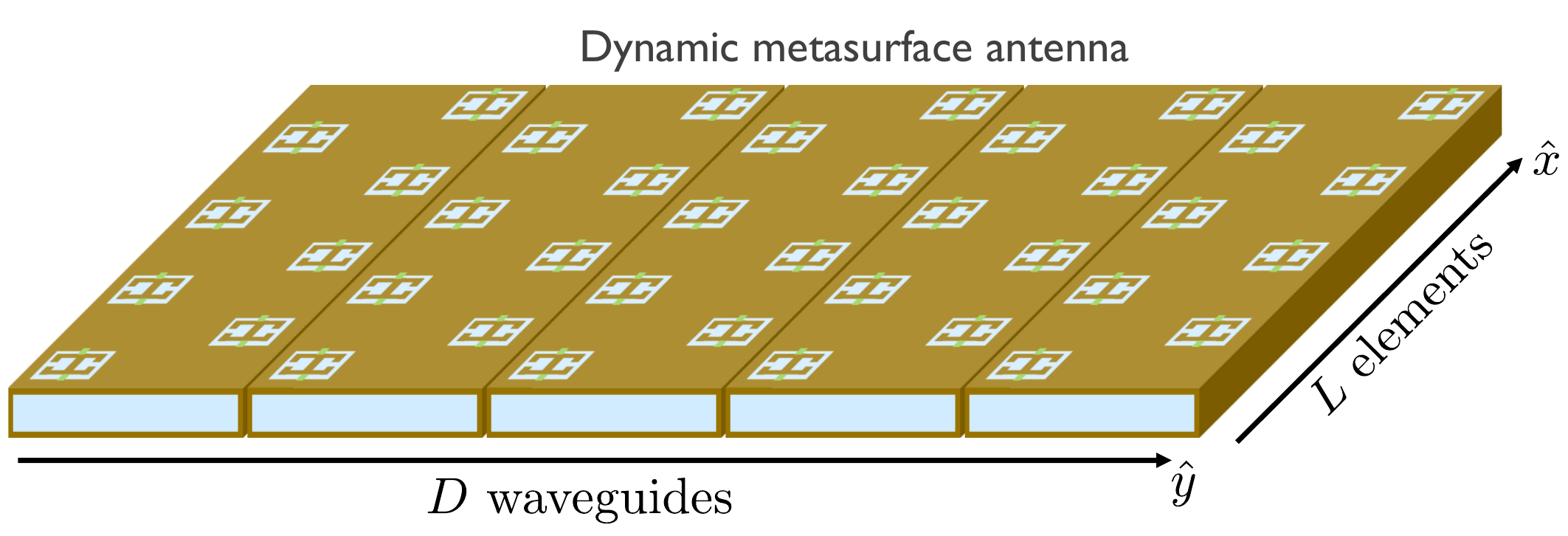}
    \caption{Visualization of an entire DMA with $D$ waveguides and $L$ radiating slot elements per waveguide.}
    \label{fig: dma array}
\end{figure}

\subsection{DMA element model}

We now discuss the general operation of the DMA to integrate its characteristics into the signal model. DMAs operate as a leaky-wave slot antenna, where each antenna element consists of a waveguide slot that radiates power \cite{SmithEtAlAnalysisWaveguideFedMetasurfaceAntenna2017}. The integration of reconfigurable components is the key enabling technology for energy-efficient beamforming with DMAs. The reconfigurable components operate by changing the resonant response of the individual DMA elements, which shifts their radiated amplitude and phase \cite{ShlezingerEtAlDynamicMetasurfaceAntennasUplink2019}. Changing the resonance of the DMA element requires altering its capacitance and/or inductance. This can be achieved in a variety of ways. Examples include integrating a varactor diode to tune the DMA element capacitance, or using a liquid crystal to tune the effective permittivity of the substrate \cite{Perez-AdanEtAlIntelligentReflectiveSurfacesWireless2021}. In our work, we focus on the varactor diode implementation. Using varactor diodes is consistent with the physical DMA models in \cite{BoyarskyEtAlElectronicallySteeredMetasurfaceAntenna2021,SmithEtAlAnalysisWaveguideFedMetasurfaceAntenna2017}. The varactor diode capacitance range limits the resonant frequency tuning of the DMA. Adjusting the DMA element resonant frequency through varactor diodes leads to a controllable radiated phase and amplitude, enabling beamforming for the DMA.

We now examine how to integrate the DMA element models in \cite{SmithEtAlAnalysisWaveguideFedMetasurfaceAntenna2017,BoyarskyEtAlElectronicallySteeredMetasurfaceAntenna2021} into a DMA precoding framework. DMA elements are typically modeled as magnetic dipole radiators whose magnetic polarizability $\am$ dictates the radiated amplitude and phase \cite{BoyarskyEtAlElectronicallySteeredMetasurfaceAntenna2021,SmithEtAlAnalysisWaveguideFedMetasurfaceAntenna2017}. As DMA element geometry mimics a complimentary electric-LC resonator, we can relate the DMA element frequency response to a simple LC circuit. The inductance and capacitance for the circuit are determined by the element geometry and the varactor diode value. For a DMA element, we let $F$ be the coupling factor and $\gamma$ be the damping factor. We also define $\omega$ as the angular frequency and $\omega_0$ as the resonant angular frequency. The LC circuit model describes the magnetic polarizability as \cite{BoyarskyEtAlElectronicallySteeredMetasurfaceAntenna2021}

\begin{equation}\label{eq: polarizability}
    \am(\omega) = \frac{F\omega^2}{\omega_0^2-\omega^2+\sfj \omega\gamma}.
\end{equation}

\noindent The varactor diode can alter \eqref{eq: polarizability} by changing the resonant frequency $\omega_0$, altering the amplitude and phase of the polarizability at a specific frequency.

Next, we discuss how to translate the magnetic dipole model to beamforming weights. For a DMA, the radiated fields are related by the dipole moment $\eta$ of each DMA element \cite{BoyarskyEtAlElectronicallySteeredMetasurfaceAntenna2021}. For the magnetic polarizability $\am(\omega)$ and waveguide magnetic field $m$, the dipole moment is given as $\eta=\am(\omega) m$ \cite{BoyarskyEtAlElectronicallySteeredMetasurfaceAntenna2021}. Since the waveguide magnetic field is fixed, tuning $\am(\omega)$ via the reconfigurable component allows for control of the radiated field through the DMA element dipole moment. Therefore, the tunable magnitude and phase of the magnetic polarizability can be thought of as the beamforming weights for a DMA precoding architecture. The integration of a varactor diode works to tune the resonant frequency $\omega_0$ of the magnetic polarizability in \eqref{eq: polarizability}. At a center angular frequency $\omega_{\sf{c}}$, these tuning states can create different radiated amplitude and phases with $\am( \omega_{\sf{c}})$, where $\omega_0$ is controlled by the varactor diode. We can normalize $\am(\omega_{\sf{c}})$ by the factor $\frac{F\omega_0}{\gamma}$ to model the available DMA beamforming weights as Lorentzian-constrained weights given by \cite{HuangEtAlJointMicrostripSelectionBeamforming2022,BowenEtAlOptimizingPolarizabilityDistributionsMetasurface2022}

\begin{equation}\label{eq: DMA weight dist}
    \mathcal{Q} = \left\{-\frac{\sfj+e^{\sfj\varphi}}{2} : \varphi\in[0,2\pi]\right\}.
\end{equation}

\noindent The normalization by $\frac{F\omega_0}{\gamma}$ is such that the highest magnitude weight value coincides with desired weights at $-1\sfj$. The normalization makes it easier to compare the DMA weights to conventional unit-norm beamforming vectors and enables the mapping methods developed in \cite{BowenEtAlOptimizingPolarizabilityDistributionsMetasurface2022}, as discussed in Section \ref{sec: DMA mapping}. Given the limited DMA weights, we define the feasible DMA beamforming set $\mathcal{F}_{\sf{DMA}}=\left\{ \mathbf{f}=\left[ f_1,\ldots,f_{\Nt} \right]^T : f_i \in \mathcal{Q}, i=1,\ldots,\Nt \right\}$ as the set of all vectors of length $\Nt$ whose elements satisfy the constraint for $\mathcal{Q}$.

The DMA waveguide affects signal propagation, so we incorporate its effect in the channel model. The electric and magnetic fields within the waveguide have a sinusoidal distribution that also effects the radiation characteristics of each DMA element. As the electromagnetic fields propagate down the waveguide, the reference wave becomes attenuated due to the radiation of energy through the DMA slots and a waveguide attenuation constant. For our signal model, however, we assume there to be no attenuation to simplify the precoding challenges \cite{HuangEtAlJointMicrostripSelectionBeamforming2022,ShlezingerEtAlDynamicMetasurfaceAntennasUplink2019,YouEtAlEnergyEfficiencyMaximizationMassive2022,BowenEtAlOptimizingPolarizabilityDistributionsMetasurface2022}. For a waveguide propagation constant $\beta$, element position $x_{d,\ell}$ representing the $d$th waveguide and $\ell$th element, and assuming a unit-amplitude reference wave, the waveguide magnetic field at each DMA element with no attenuation is modeled as \cite{HuangEtAlJointMicrostripSelectionBeamforming2022}

\begin{equation}\label{eq: waveguide}
    \mbl = e^{-\sfj \beta x_{d,\ell}}.
\end{equation}

\noindent We define $\mathbf{m} \in \mathbb{C}^{DL\times 1}$ as the vectorized waveguide magnetic field values. We assume that the bandwidth is narrow enough such that the propagation constant $\beta$ does not change significantly with frequency, allowing us to approximate $\mathbf{m}$ as frequency-flat. The propagation constant of the DMA waveguide introduces an inherent phase advance into the DMA system. Therefore, the DMA weights must be tuned both to steer the beam in a desired direction, and to counteract this phase advance.

The Lorentzian-constrained weights and waveguide model highlight the main tradeoff for enabling wireless communications with DMAs. DMAs offer an energy-efficient solution to beamforming via low-power reconfigurable components. The limited weights and grating lobes from the waveguide phase advance, however, often lead to decreased beamforming gain and beamsteering accuracy compared to phased arrays \cite{BoyarskyEtAlGratingLobeSuppressionMetasurface2020,BoyarskyEtAlElectronicallySteeredMetasurfaceAntenna2021}. In the next section, we will integrate these parameters into a MISO signal model.

\subsection{Communication and signal model}

Next, we describe the MISO precoding architecture for implementation with a DMA. We first define the single-stream MISO-OFDM received signal $y[k]$ for Gaussian-distributed noise $n[k] \sim \mathcal{N}_{\mathbb{C}}(0,1)$, wireless channel $\mathbf{h}[k]$ from transmit antenna $\nt$ to the single receiver, beamformer $\mathbf{f}$, transmitted symbol $s[k]$ and signal-to-noise ratio (SNR) $\rho$ as \cite{ParkEtAlDynamicSubarraysHybridPrecoding2017}

\begin{equation}
    y[k] = \sqrt{\rho} \mathbf{h}^*[k] \mathbf{f}s[k] + n[k]. 
\end{equation}

\noindent In the DMA architecture, the beamformer $\mathbf{f}$ will consist of the DMA beamforming weights, whose frequency response can be taken as a bandpass filter given the LC circuit model. It has been shown that the quality factor for the DMA element frequency response typically results in an approximately equal gain across frequency for a narrowband waveguide excitation \cite{ShlezingerEtAlDynamicMetasurfaceAntennasUplink2019}. Therefore, we assume that the bandwidth is small enough such that we can approximate the beamformer $\mathbf{f}$ as frequency-flat. As our goal is to investigate the use of DMAs for codebook design, we assume $\mathbf{f}$  to belong to a set of codewords in a codebook $\mathcal{F} \subseteq \mathcal{F}_{\sf{DMA}}$ for beamforming. For the codebook  $\mathcal{F}$, the spectral efficiency optimized by the codeword selection assuming a known channel is then given as \cite{JrLozanoFoundationsMIMOCommunication2018}

\begin{equation}\label{eq: spec eff}
    C\left(\rho|\{\mathbf{h}[k]\}_{k=1}^K \right) = \argmax\limits_{\mathbf{f} \in \mathcal{F}}\frac{1}{K}\sum\limits_{k=1}^{K}\log_2{\left(1+\rho|\mathbf{h}^*[k]\mathbf{f}|^2 \right)}.
\end{equation}

\noindent We choose to investigate codebook designs for DMAs as codebooks are practical, low-complexity, and widely used in current standards 
\cite{ZhouOhashiEfficientCodebookbasedMIMOBeamforming2012,UwaechiaMahyuddinComprehensiveSurveyMillimeterWave2020}. This topic also has not been studied in literature with DMAs like \cite{HuangEtAlJointMicrostripSelectionBeamforming2022,YouEtAlEnergyEfficiencyMaximizationMassive2022,ShlezingerEtAlDynamicMetasurfaceAntennasUplink2019}. Therefore, we implement a hierarchical codebook design with DMAs in this paper to help close the gap in literature and compare against current beamforming methods with phased arrays.

We compare a typical analog precoding architecture to a DMA precoding architecture in Fig. \ref{fig: precoding architectures} for a single-stream scenario. Using analog beamforming with phase shifters, an $\Nt \times 1$ beamformer $\mathbf{\fps}$ contains weight information to control each antenna. For a DMA beamformer $\fdma$, however, we use the integrated varactor diode to alter the radiated amplitude and phase instead of phase shifters. The DMA beamforming weight information is translated to the varactor diode voltage control for beamforming via a weight-to-voltage mapping. The elements of the DMA beamformer $\fdma$ are restricted by the Lorentzian-constrained weights $\mathcal{Q}$, while the elements of the phased array beamformer $\fps$ lie on the complex unit circle.

We also must integrate the effects of the DMA waveguide into the wireless channel for the DMA. The DMA waveguide introduces an additional phase term due to the phase advance of the waveguide field propagation for the $d$th waveguide and $\ell$th element as $\mbl = e^{-\sfj \beta x_{d,\ell}}$. The input signal to the DMA then propagates through the DMA waveguide and is radiated out through the DMA elements to become a part of the wireless channel for this DMA element $\hbl[k]$, yielding $\hbl[k] \mbl$. To account for this effect for all DMA elements, we multiply the wireless channel $\mathbf{h}[k]$ by its corresponding phase advance $\mathbf{m}$ to result in the effective channel

\begin{equation}
    \mathbf{\heff}[k] = \mathbf{h}[k] \odot \mathbf{m}.
\end{equation}

\noindent Therefore, the DMA element weights need to be set both to steer the beam and counteract the waveguide phase advance to create a high-gain beam in a desired direction. In the next section, we examine methods for setting the DMA weights developed in prior work \cite{SmithEtAlAnalysisWaveguideFedMetasurfaceAntenna2017,BowenEtAlOptimizingPolarizabilityDistributionsMetasurface2022}.

% was .85
\begin{figure} 
    \centering
  \subfloat[\label{fig: ant array precoding}]{%
       \includegraphics[width=.48\linewidth]{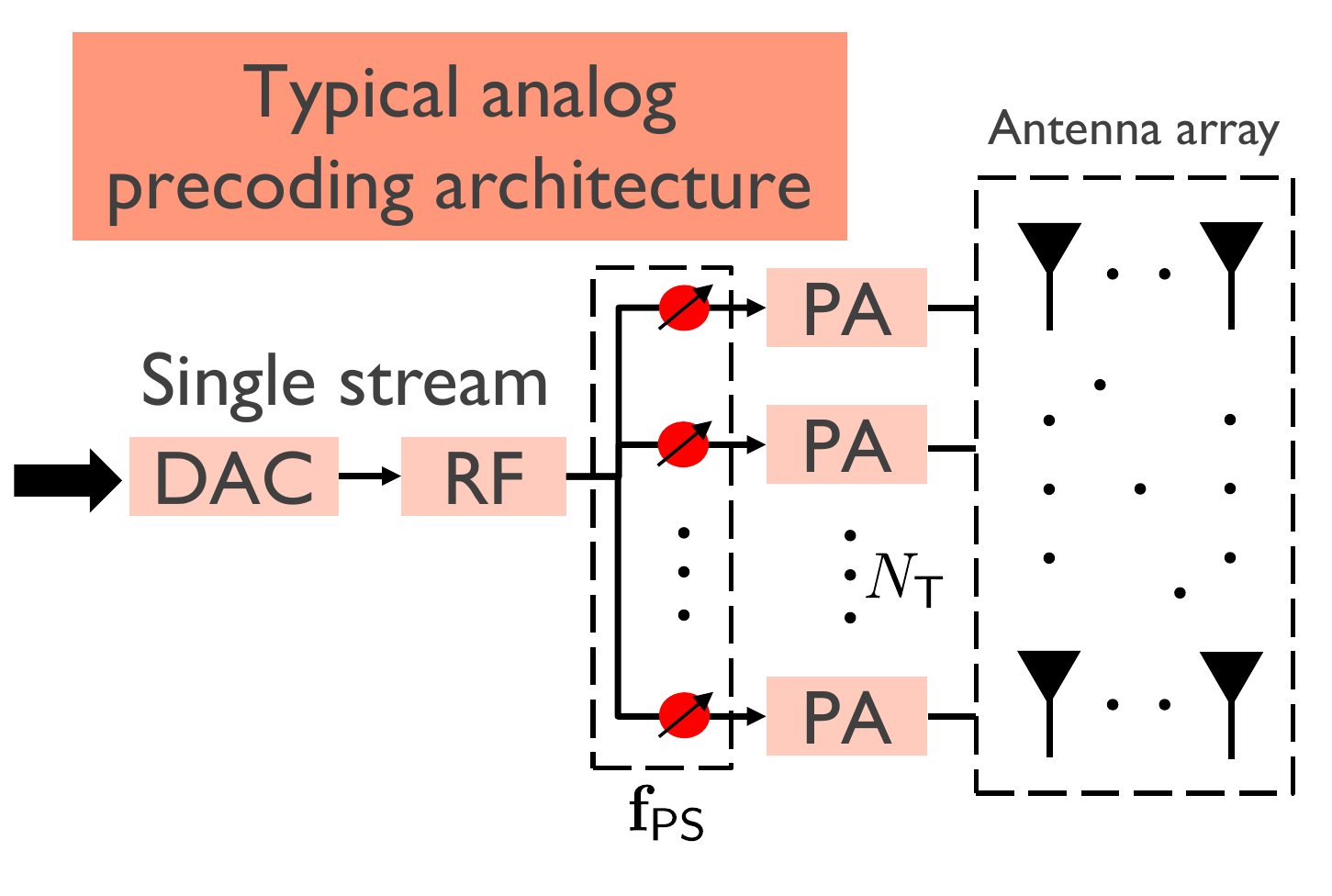}}
    \hfill
  \subfloat[\label{fig: dma precoding}]{%
        \includegraphics[width=.48\linewidth]{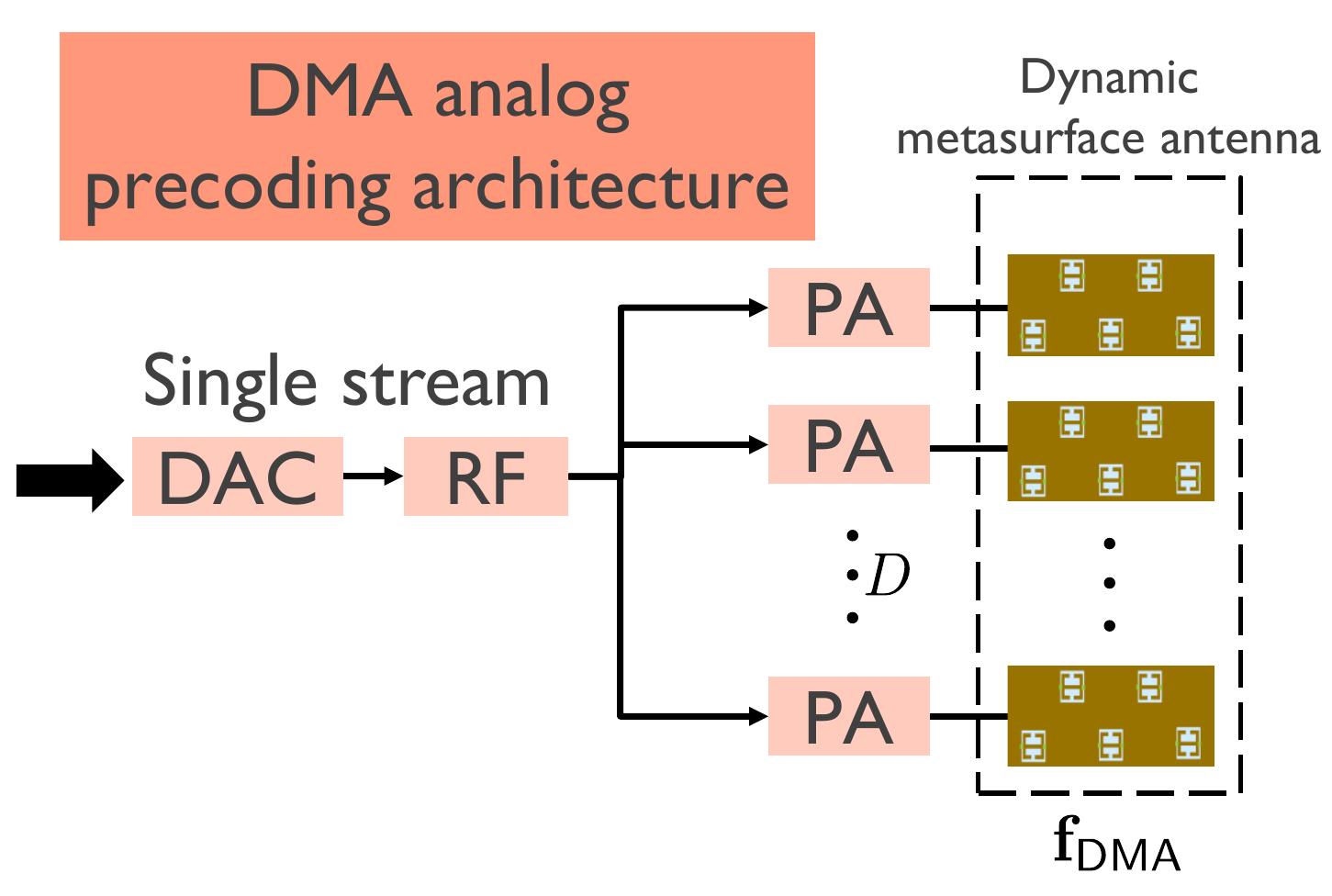}}

  \caption{Analog beamforming transmitter architecture for (a) a typical phased antenna array and (b) the DMA. Typical analog implementations enable precoding with phase shifters, while DMAs use reconfigurable components.}
  \label{fig: precoding architectures} 
\end{figure}

\section{DMA mapping design}\label{sec: DMA mapping}

We now discuss mapping techniques that enable DMA beamforming. We formulate the DMA mapping problem by defining each DMA element to have a specified weight $\qbl$, where the weight values are limited by the Lorentzian constraint $\qbl \in \mathcal{Q}$. The desired weights $\wbl$ are assumed to lie on the complex unit circle so that $|\wbl|=1$. These weights mimic a typical phased array implemented with unit-amplitude phase shifters. Mapping then takes the desired weight $\wbl$ and finds a corresponding Lorentzian-constrained weight $\qbl$ that is close to the desired weight with respect to a metric. The individual DMA weight values form the overall DMA weight matrix $\mathbf{Q}$, and the desired weights form $\mathbf{W}$. We can then define the DMA beamformer as $\fdma={\sf{vec}}(\mathbf{Q})$. In a similar fashion, we define $\mathbf{f}_{\sf{PS}} ={\sf{vec}}(\mathbf{W}) $ to represent the desired beamformer with phase shifters for beamforming. Thus, we specify the DMA analog beamformer in the same way as a typical array but with the additional Lorentzian-constrained weights. 

The DMA beamformer $\fdma$ is found by mapping the desired beamformer $\fps$ onto the DMA weights given by $\mathcal{Q}$, also accounting for the waveguide channel effects. We define the effective desired weight $\wtbl$ that counteracts the waveguide phase advance as

\begin{equation}
    \wtbl = \wbl  \mbl^*.
\end{equation}

\noindent We then use this new weight for mapping to find $\fdma$. While there are multiple mapping methods presented in literature, we focus here on the two mappings that were found to provide the highest beamforming gain and steering accuracy: Euclidean modulation and Lorentzian-constrained modulation \cite{BowenEtAlOptimizingPolarizabilityDistributionsMetasurface2022}. We summarize these mapping techniques as follows, based on the work in \cite{BowenEtAlOptimizingPolarizabilityDistributionsMetasurface2022, SmithEtAlAnalysisWaveguideFedMetasurfaceAntenna2017}. We also develop a novel method for creating larger beamforming gains through a phase rotation of the desired weights.

\subsection{DMA mapping}

Euclidean modulation is a mapping that involves minimizing the Euclidean distance between the effective desired weights and DMA weights. For the weight $\wtbl$, we can formulate the optimal weight $q^{\sf{EM}}_{d,\ell}$ for Euclidean modulation (EM) by setting \cite{BowenEtAlOptimizingPolarizabilityDistributionsMetasurface2022}

\begin{equation}\label{eq: EM}
    q^{\sf{EM}}_{d,\ell} = \argmin\limits_{\qbl \in \mQ}|\qbl-\wtbl|.
\end{equation}

\noindent This can be determined through a brute-force search algorithm, or formulaically as described in the following lemma: 

\begin{lemma}
For the desired weight $\wtbl$, the weight $q_{d,\ell} \in \mathcal{Q}$ that minimizes $|\qbl-\wtbl|$ is given by $q^{\sf{EM}}_{d,\ell} = -\frac{\sfj +e^{\sfj ( \angle(2\wtbl+\sfj)-\pi )}}{2}$.

\end{lemma}

\begin{proof}\renewcommand{\qedsymbol}{}
    See Appendix.
\end{proof}

\noindent Under ideal conditions that assume no waveguide attenuation and a theoretical DMA model, prior work has shown that Euclidean modulation suffers from increased grating lobes and lower directivity due to the non-uniform phase differences between elements \cite{BowenEtAlOptimizingPolarizabilityDistributionsMetasurface2022}.

Lorentzian-constrained modulation allows for the DMA weights to be determined in a formulaic manner similar to the array factor. The desired weight phase $\angle \wtbl$ is substituted into the Lorentzian-constrained weight equation to calculate the Lorentzian-constrained (LC) modulation weights as \cite{BowenEtAlOptimizingPolarizabilityDistributionsMetasurface2022}

\begin{equation}\label{eq: LC}
    \qbl^{\sf{LC}}=-\frac{\sfj-e^{\sfj \angle \wtbl}}{2}.
\end{equation}

\noindent The Lorentzian-constrained modulation provides a nearly constant phase difference between DMA elements as it minimizes the distance between the weight phase values \cite{BowenEtAlOptimizingPolarizabilityDistributionsMetasurface2022}, whereas Euclidean modulation minimizes the distance between desired and Lorentzian-constrained weights. Further details regarding the Euclidean and Lorentzian-constrained modulation techniques can be found in \cite{BowenEtAlOptimizingPolarizabilityDistributionsMetasurface2022}.

While Lorentzian-constrained modulation provides greater beamforming gain and accuracy than Euclidean modulation \cite{BowenEtAlOptimizingPolarizabilityDistributionsMetasurface2022}, there is still limited work comparing these mapping techniques with realistic DMA designs. Moreover, these mappings techniques have not been optimized when considering non-ideal DMA characteristics, such as element perturbation of the waveguide. We will develop an optimization method to maximize the theoretical beamforming gain through a phase rotation of the desired weights in the following section.

\subsection{Mapping optimization}

We now introduce a novel technique of applying a phase rotation to the desired precoder $\fps$ to maximize beamforming gain with current mapping techniques. Normally, adding an arbitrary phase rotation $\zeta$ to the antenna weights does not change the final beamforming gain or radiation pattern so that $\fps \equiv \fps e^{\sfj \zeta} \; \forall \; \zeta \in [0,2\pi)$, where $\equiv$ indicates that the two quantities are equivalent in terms of beamforming gain. Applying a phase rotation to the desired weights, though, will result in a mapping onto different Lorentzian-constrained weights. We demonstrate the phase rotation mapping by decomposing the Lorentzian-constrained modulation with a phase rotation into its real and imaginary parts as $\qbl^{\sf{LC}}(\zeta)=\frac{1}{2} \left[ \cos \left(\angle(\wtbl e^{\sfj \zeta}) \right) +{\sfj} \left( \sin \left(\angle(\wtbl e^{\sfj \zeta}) \right) -1 \right) \right] $. The magnitude and phase for the Lorentzian-constrained weight are then given by

\begin{equation}\label{eq: zeta amp}
    \begin{split}
        |\qbl^{\sf{LC}}(\zeta)| & = \sqrt{\frac{1}{4} \left[ \cos \left(\angle(\wtbl e^{\sfj \zeta}) \right)^2 + \left( \sin \left(\angle(\wtbl e^{\sfj \zeta}) \right) -1 \right)^2 \right] } \\ & = \sqrt{\frac{1-\sin \left(\angle(\wtbl e^{\sfj \zeta}) \right)}{2}},
    \end{split}
\end{equation}

\begin{equation}\label{eq: zeta phase}
    \angle \qbl^{\sf{LC}}(\zeta) = \arctan\left( \frac{ \sin \left(\angle(\wtbl e^{\sfj \zeta}) -1 \right)}{\cos \left(\angle(\wtbl e^{\sfj \zeta}) \right)} \right).
\end{equation}

\noindent We find in \eqref{eq: zeta amp} and \eqref{eq: zeta phase} that each phase rotation $\zeta$ leads to a unique complex value for the resulting DMA weight magnitude and phase. Therefore, the set of mapped weights from a phase rotation will correspond to a new radiation pattern with differences in the beamwidth, beamforming gain and steering accuracy. We further illustrate this phase rotation in Fig. \ref{fig: dma phase rotation}. 

To achieve larger beamforming gain, we develop an approach for optimizing the phase rotation of the desired weights directly through electromagnetic simulations. While studies have shown that Euclidean and Lorentzian-constrained modulation provide accurate beamsteering through simulation and experimental results \cite{SmithEtAlAnalysisWaveguideFedMetasurfaceAntenna2017,BoyarskyEtAlElectronicallySteeredMetasurfaceAntenna2021}, no studies have been conducted to analyze their accuracy in calculating beamforming gain. Moreover, the ideal waveguide channel model implemented in \eqref{eq: waveguide} assumes negligible waveguide field attenuation due to the waveguide attenuation constant and power leaked out through the DMA elements. As the radiated fields from the DMA element are proportional to the waveguide fields, the waveguide field attenuation will certainly impact the resulting DMA beam pattern gain. We leave the development of a sophisticated beamforming gain model for future work, and instead use a brute-force search with electromagnetic simulations of a DMA. We restrict the phase rotation $\zeta$ to be a part of the set $\mathcal{Z} = \left\{ \frac{2\pi (0)}{Z},\frac{2\pi (1)}{Z},\ldots, \frac{2\pi (Z-2)}{Z}, \frac{2\pi (Z-1)}{Z} \right\}$, which describes the $Z$ discrete phase rotations possible. For all $\zeta \in \mathcal{Z}$, we calculate the LC and EM mapped weights as $\qbl^{\sf{LC}}(\zeta)=-\frac{\sfj-e^{\sfj \angle(\wtbl e^{\sfj \zeta})}}{2}$, $q^{\sf{EM}}_{d,\ell}(\zeta) = -\frac{\sfj+e^{\sfj ( \angle(2\wtbl e^{{\sfj}\zeta}+\sfj)-\pi ) }}{2}$. We then perform electromagnetic simulations to extract the resulting DMA beam pattern gain as a function of its azimuth angle $\phi$, denoted as $v(\zeta,\phi)$. At the desired beamsteering angle $\phi_0$, we find the phase rotation that maximizes beamforming gain as $\hat{\zeta}(\phi_0) = \argmax\limits_{\zeta \in \mathcal{Z}} v(\zeta,\phi_0)$. The optimized DMA weights are then given by

\begin{equation}
    \hat{q}_{d,\ell}^{\sf{LC}}(\phi_0)=-\frac{\sfj-e^{\sfj \angle(\wtbl e^{\sfj\hat{\zeta}(\phi_0)})}}{2},
\end{equation}

\begin{equation}
    \hat{q}^{\sf{EM}}_{d,\ell}(\phi_0) = -\frac{\sfj+e^{\sfj ( \angle(2\wtbl e^{{\sfj}\hat{\zeta}(\phi_0)}+\sfj)-\pi ) }}{2}.
\end{equation}

\noindent We implement this optimization to compare against the typical mapping techniques and enhance system performance.

The brute-force phase rotation optimization provides an additional methodology for enhancing the beamforming gain of DMAs. The ultimate objective for developing the phase rotation optimization and mapping techniques is to analyze their effectiveness in a codebook design with a realistic DMA. Thus, we now discuss the design of the DMA and its simulated radiation characteristics.

% was 1
\begin{figure}
    \centering
    \includegraphics[width=.42\linewidth]{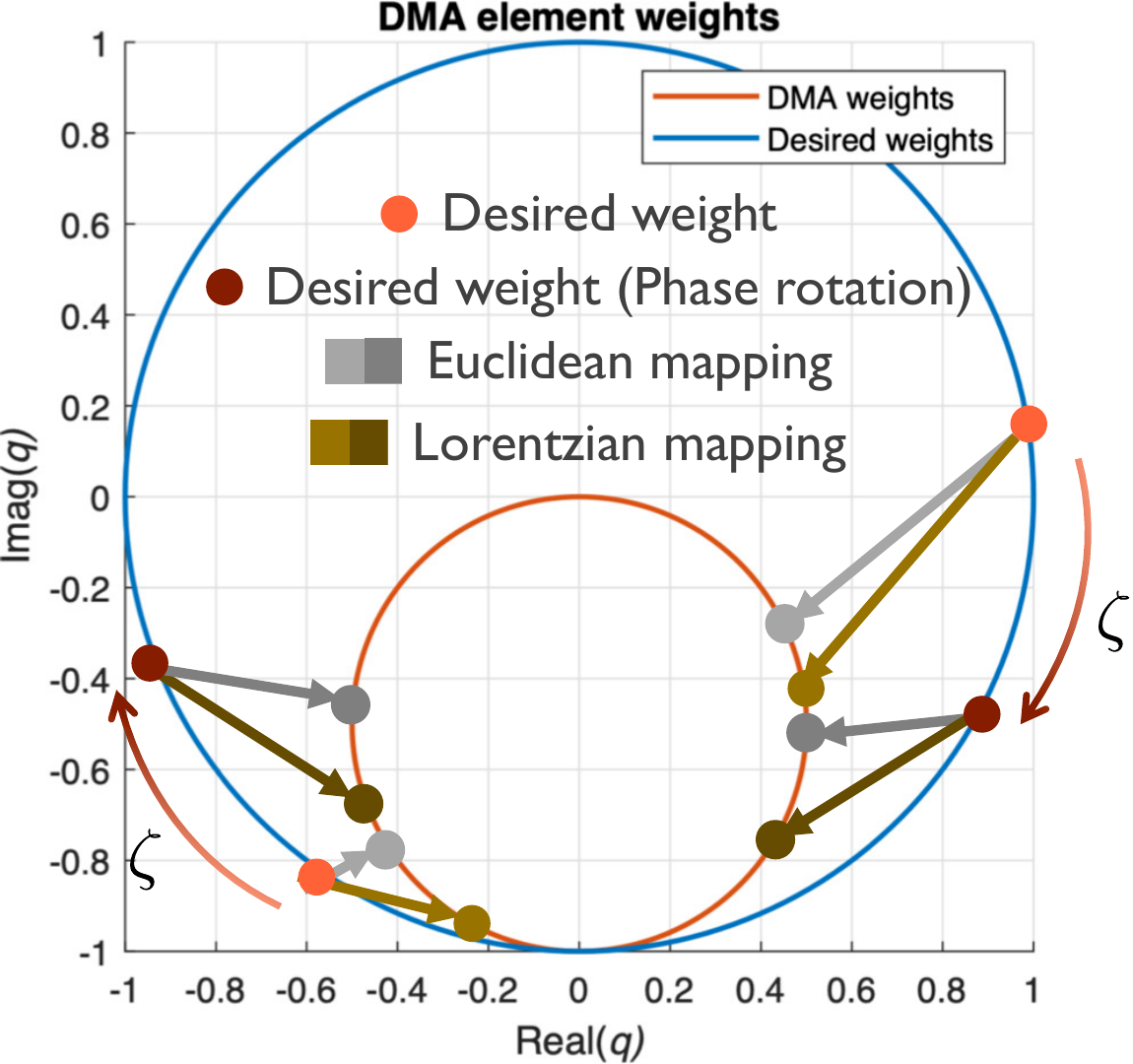}
    \caption{Applying a phase rotation to the desired weights results in a new set of mapped weights onto the Lorentzian-constrained weights. We can leverage this to achieve greater beamforming gain with the DMA.}
    \label{fig: dma phase rotation}
\end{figure}

\section{DMA design and hardware implementation}\label{sec: DMA design and hardware}

Although DMAs are motivated as an energy-efficient solution to large, MIMO arrays, there is limited work modeling their total power consumption and incorporating realistic DMA designs into the precoding architecture. In this section, we develop a DMA design based on previous work that we will use to simulate codebook beam patterns. Furthermore, we establish a transmitter power consumption model for both the DMA and a typical phased antenna array for energy efficiency comparisons.

\subsection{DMA design}

We now discuss the design of the DMA and details of the simulation process. Since the purpose of this paper is to analyze codebook designs for DMAs, we choose to modify a previously-established DMA rather than develop our own. This allows for an analysis of the proposed DMA system model with DMAs that have been tested experimentally. We leave the joint optimization of a new DMA design and the proposed system model for future work. 

We use and modify the DMA design from \cite{BoyarskyEtAlElectronicallySteeredMetasurfaceAntenna2021}. The element geometry for this design is depicted in Fig. \ref{fig: dma array}, consisting of a complimentary electric-LC resonator \cite{BoyarskyEtAlElectronicallySteeredMetasurfaceAntenna2021}. We use a hollow waveguide design to increase the DMA radiation efficiency. Each antenna element is loaded with two controllable varactor diodes. We choose to analyze a DMA with $L=8$ elements per waveguide, and $D=5$ waveguides as it fits well with the binary hierarchical codebook we will discuss shortly. We can then design the DMA element to meet a radiated power performance requirement given the $D \times L$ array. With $L=8$ DMA elements per waveguide, we need to ensure that the DMA is capable of radiating out all the input power from the excited waveguide. Power not radiated out through the DMA elements will be absorbed at the end of the waveguide, resulting in low radiation efficiency. Additionally, radiating too much power through the first initial elements will cause the end elements to radiate nearly no power. It is important to design the DMA such that most or all of the elements are radiating out power to ensure an efficient antenna. We tune the amount of radiated power associated with each Lorentzian-constrained weight in the DMA element design via its geometry. We also offset the DMA elements from the center of the waveguide as another method to control the radiated power \cite{BoyarskyEtAlElectronicallySteeredMetasurfaceAntenna2021}. We alter the design in \cite{BoyarskyEtAlElectronicallySteeredMetasurfaceAntenna2021} to achieve a maximum element radiated power of around 40\% of the input power through trial and error, corresponding to the Lorentzian-constrained weight of $-1\sfj$. Table \ref{MSA table params} shows important specifications for the designed DMA, and we leave further details of the DMA design process and specific features of this DMA to be found in \cite{BoyarskyEtAlElectronicallySteeredMetasurfaceAntenna2021}.

We model and simulate the designed DMA in the electromagnetic simulation software HFSS. We represent the varactor diodes as lumped elements and change the capacitance to vary the resonant frequency of the DMA elements. After simulating the characteristics of a single DMA element, we calculate the magnetic polarizability associated with each varactor diode value through the resulting S-parameters, waveguide dimensions $a,b$, and waveguide propagation constant $\beta$ as \cite{BoyarskyEtAlElectronicallySteeredMetasurfaceAntenna2021}

\begin{equation}\label{eq: am sim}
    \am(\omega) = -\frac{\sfj ab}{\beta}\left(1+S_{11}(\omega)-S_{21}(\omega)\right).
\end{equation}

\noindent Next, we normalize the simulated $\am$ values such that the highest magnitude weight value corresponds with the weight $-1\sfj$. This is equivalent to normalizing by the factor $\frac{F\omega_0}{\gamma}$ discussed in Section \ref{sec: system model}. We show the resulting weights in Fig. \ref{fig: dma weight sim}. We chose discrete varactor diode values for the simulation, and interpolate the data to calculate the achievable weight-to-capacitance mapping for all varactor diode values. We find that the simulated DMA weights match very closely to the theoretical Lorentzian-constrained weights, which validates the beamforming model. Furthermore, we see that there is a small weight gap near the zero-weight for the simulated results. This is due to the realistic tuning constraints of a varactor diode, and will be further analyzed in Section \ref{sec: results} for its impact on the DMA beam patterns.

% was 1
\begin{figure}
    \centering
    \includegraphics[width=.42\linewidth]{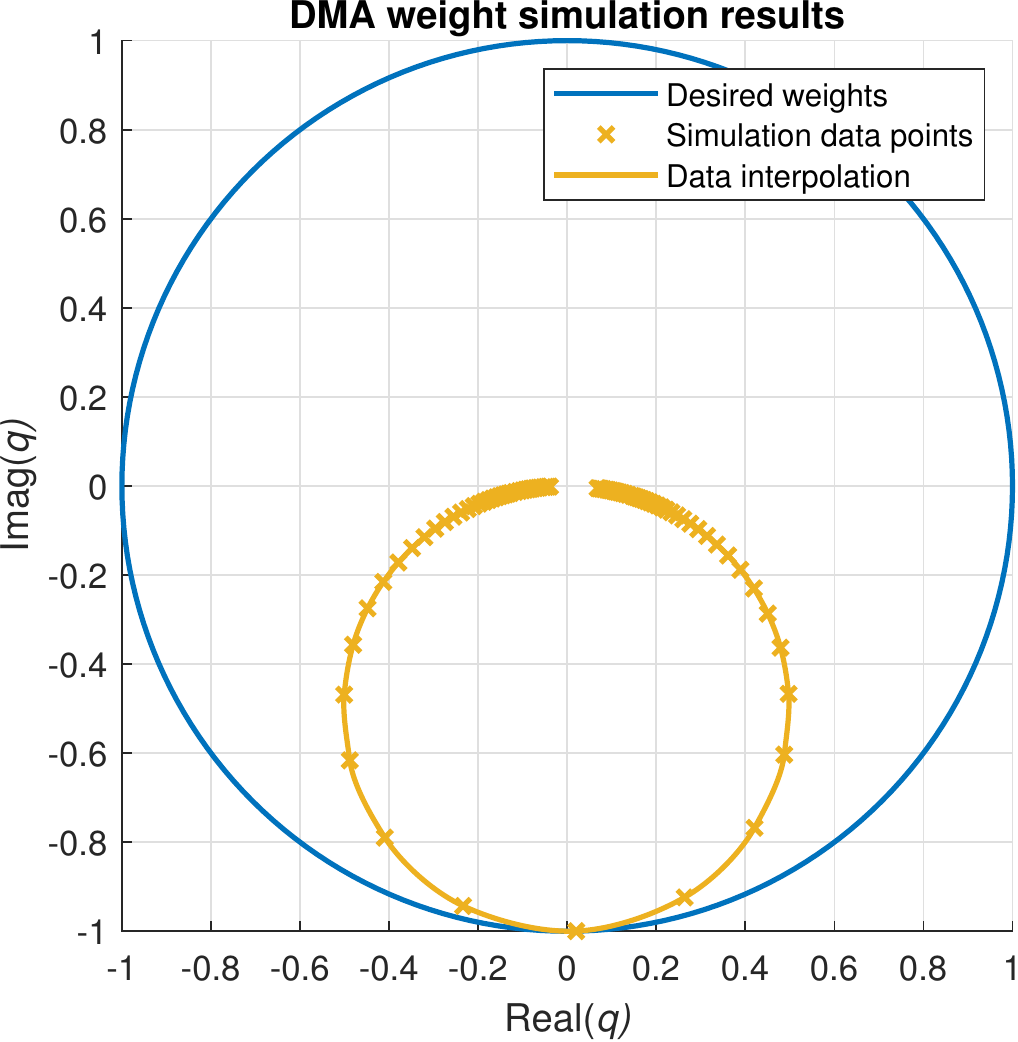}
    \caption{Simulated results for the Lorentzian-constrained weights of the designed DMA in HFSS using \eqref{eq: am sim}. We find that the DMA accurately achieves the Lorentzian-constrained weights defined in \eqref{eq: DMA weight dist}.}
    \label{fig: dma weight sim}
\end{figure}

\begin{table}[]
    \centering
    \caption{Parameters for the modeled DMA based on \cite{BoyarskyEtAlElectronicallySteeredMetasurfaceAntenna2021}.}
    \begin{tabular}{|c|c|c|}
        \hline
        Elements per waveguide & $L$ & $8$ \\
        \hline
        Number of waveguides & $D$ & $5$ \\
        \hline
        Element spacing in $\x$ direction & $d_{\sf{x}}$ & $5$ mm \\ \hline
        Element spacing in $\y$ direction & $d_{\sf{y}}$ & $10$ mm \\ \hline
        Operating wavelength & $\lambda$ & $20$ mm \\ \hline
        Operating frequency & $f_0$ & $15$ GHz \\ \hline
    \end{tabular}
    \label{MSA table params}
\end{table}

\subsection{Power consumption and component loss model}\label{sec: power cons}

We define power consumption models for the DMA and typical phased antenna array. We analyze scenarios for an analog MISO architecture as described in Section \ref{sec: system model}. We use these power consumption models to analyze the energy efficiency of the DMA compared to a phased array. Fig. \ref{fig: precoding architectures} shows the transmitter architecture for the phased array and DMA \cite{RibeiroEtAlEnergyEfficiencyMmWaveMassive2018}. Here, we assume a fully-connected architecture so that each antenna element is connected to its own phase shifter \cite{RialEtAlHybridMIMOArchitecturesMillimeter2016}.

Next, we define a power consumption model for the two beamforming cases. We denote the power consumption for the individual components in the transmitter architecture as follows: let $P_{\sf{LO}}$ for the local oscillator, $P_{\sf{PA}}$ for all power amplifiers, $P_{\sf{DAC}}$ for the digital-to-analog converter as a function of the bit resolution $b_{\sf{DAC}}$ and sampling frequency $F_{\sf{s}}$, $P_{\sf{PS}}$ for the phase shifter, and $P_{\sf{VAR}}$ for the varactor diode. We assume a single-stream scenario with only one RF chain. Putting these components together, we define the power consumption model for the typical phased array $P_{\sf{ANT}}$ and DMA  $P_{\sf{DMA}}$ as \cite{RibeiroEtAlEnergyEfficiencyMmWaveMassive2018}

\begin{equation}
    P_{\sf{ANT}} = P_{\sf{LO}} +P_{\sf{PA}}+2P_{\sf{DAC}}(b_{\sf{DAC}},F_{\sf{s}})+P_{\sf{RF}}  +\Nt  P_{\sf{PS}},
\end{equation}

\begin{equation}
    P_{\sf{DMA}} = P_{\sf{LO}} +P_{\sf{PA}}+2P_{\sf{DAC}}(b_{\sf{DAC}},F_s)+P_{\sf{RF}}+\Nt P_{\sf{VAR}}.
\end{equation}

\noindent Additional information regarding this power consumption model can be found in \cite{RibeiroEtAlEnergyEfficiencyMmWaveMassive2018}. We list the specific power consumption values used for each component in Table \ref{table: power consumption} \cite{RibeiroEtAlEnergyEfficiencyMmWaveMassive2018}. It is important to note that the power amplifier consumption is dependent on the actual transmit power $P_{\sf{T}}$ as $P_{\sf{PA}} = \frac{P_{\sf{T}}}{\eta_{\sf{PA}}}$, where $\eta_{\sf{PA}}$ is the power amplifier efficiency. We also note that the difference in these power consumption models is the phase shifter versus varactor diode power consumption.

We integrate the loss of the RF components through a reduction in the transmit power from the initial input power $P_{\sf{IN}}$. As described in \cite{RibeiroEtAlEnergyEfficiencyMmWaveMassive2018}, we implement the total loss through an $N_{\sf{A}}$-way power divider with static loss $\bar{L}_{\sf{D}}$ as $L_{\sf{D}}=\bar{L}_{\sf{D}}\lceil{\log_2{N_{A}}}\rceil$. We also denote the static loss due to phase shifters for each element by $L_{\sf{PS}}$. The system transmit power can then be defined as 

\begin{equation}\label{eq: trans power}
    P_{\sf{T}} = \frac{P_{\sf{IN}}}{L_{\sf{D}}L_{\sf{PS}}}.
\end{equation}

\noindent For the phased array, the power divider will have $N_{\sf{A}} = \frac{\Nt}{2}$ (see Section \ref{sec: results}),  while for the DMA the power divider will have $N_{\sf{A}} = D<\frac{\Nt}{2}$. Moreover, with no phase shifters attached to the DMA, there will be no loss from $L_{\sf{PS}}$ for the DMA case, and we assume a lossless waveguide excitation component. We analyze the effects of this transmit power loss in Section \ref{sec: results}.

Lastly, as the goal of this work is to analyze the energy efficiency of the DMA compared to the phased array, we define the energy efficiency as a function of the spectral efficiency in \eqref{eq: spec eff} for the DMA $C_{\sf{DMA}}$ and phased array  $C_{\sf{ANT}}$ as \cite{YouEtAlEnergyEfficiencyMaximizationMassive2022}

\begin{equation}\label{eq: EE msa}
    \eta_{\sf{DMA}} = \frac{C_{\sf{DMA}}}{P_{\sf{DMA}}},
\end{equation}

\begin{equation}\label{eq: EE sba}
    \eta_{\sf{ANT}} = \frac{C_{\sf{ANT}}}{P_{\sf{ANT}}}. 
\end{equation}

\noindent Comparing the energy efficiencies for the two scenarios provides a metric to analyze the tradeoff between the lower DMA power consumption and larger achievable beamforming gains for the phased array. We perform simulations to determine this tradeoff in Section \ref{sec: results}.

\begin{table}[]
    \centering
    \caption{Power consumption and component loss in the transmitter architecture model.}
    \begin{tabular}{|c|c|c|}
        \hline
        \textbf{Component} & \textbf{Notation} & \textbf{Value} \\ \hline
    
        Power amplifier efficiency & $\eta_{\sf{PA}}$ & 27\% \\
        Power amplifier & $P_{\sf{PA}}$ & $P_{\sf{PA}} = \frac{P_{\sf{T}}}{\eta_{\sf{PA}}}$ \\
        Phase shifter (active ; passive) & $P_{\sf{PS}}$ & 21.6 ; 0 mW \\
        DAC & $P_{\sf{DAC}}$ & 75.8 mW\\
        Local oscillator & $P_{\sf{LO}}$ & 22.5 mW \\
        RF chain & $P_{\sf{RF}}$ & 31.6 mW \\ 
        Varactor diode & $P_{\sf{VAR}}$ & 0 mW \\ \hline
        Phase shifter (active ; passive) & $L_{\sf{PS}}$ & -2.3 ; 8.8 dB \\
        Two-way power divider & $\bar{L}_{\sf{D}}$ & 0.6 dB \\ 
         \hline
        
    \end{tabular}
    \label{table: power consumption}
\end{table}

\subsection{DMA hierarchical codebook}

Lastly, we develop methods to enable codebook-based beamforming with the designed DMA and DMA analog precoding architecture. Codebook-based beamforming involves creating a discrete list of possible beam patterns. In a wireless scenario, each beam pattern in the codebook is then tested to find the beam pattern that maximizes the receive SNR. A typical codebook design is the DFT codebook, which provides sufficient angular coverage and beamforming gain for a uniform linear array. Our goal is to apply current methods of designing codebooks to the DMA analog precoding architecture with the additional constraint of the DMA weights. We design a codebook for the DMA by mapping the weights of a desired beamformer $\fps$ onto the DMA weights to create $\fdma$. We will implement Euclidean and Lorentzian-constrained modulation for the DMA weights and the proposed phase rotation optimization to compare mapping techniques. We choose to integrate a binary hierarchical DFT codebook with the DMA to leverage the unique zero-weight in the Lorentzian-constrained weights. 

A hierarchical codebook contains multiple layers of codewords that gradually generate narrower beams. Therefore, hierarchical codebooks can lead to a quicker search of the best codeword in the codebook compared to exhaustive search methods \cite{XiaoEtAlHierarchicalCodebookDesignBeamforming2016}. Wider beams are generally created by turning off certain antenna elements, which involves additional switches and control mechanisms from a system design perspective. We analyze DMAs with hierarchical codebooks here as the Lorentzian-constrained weights contain a zero-weight, so no additional switches are required to turn off antennas. Because of this, a hierarchical codebook can leverage the Lorentzian-constrained weights to generate the desired wide beams. As the first layer consists of one single beam pattern, we will analyze the mapping techniques and optimizations for all layers of the hierarchical codebook beyond the first layer.

We generate the binary hierarchical codebook based on the DEACT approach developed in \cite{XiaoEtAlHierarchicalCodebookDesignBeamforming2016}. We first define the array factor for $L$ elements, wavenumber $k$, spacing $d_{\sf{x}}$, and desired steering angle $\Omega$ as

\begin{equation}
    \mathbf{a}(L,\Omega)=\frac{1}{\sqrt{L}}\left[ e^{\sfj k d_{\sf{x}} 0\Omega}, \ldots,e^{\sfj k d_{\sf{x}}(L-1)\Omega} \right]^T.
\end{equation}

\noindent Next, we define a hierarchical codebook layer $r$ to consist of $2^{(r-1)}$ DFT codewords that adequately span angular space through the resulting beamwidths. The number of DFT beams scales with the hierarchical codebook layer, thus lower layers have fewer beams with wider beamwidths, and higher layers have many beams with narrower beamwidths. For the $r$th hierarchical codebook layer and $c$th codeword, the desired antenna weights are given by

\begin{equation}
    \mathbf{f}_{\sf{HCB}}(r,c)=\left[ \mathbf{a}\left( 2^r,-1+\frac{2c-1}{2^r} \right)^T, \mathbf{0}^T_{(L-2^r)-1} \right] ^T.
\end{equation}

\noindent For our DMA scenario of $L=8$, there will be $R=4$ hierarchical codebook layers, with the final layer having $C=L=8$ codewords.

As the DEACT approach is designed for uniform linear arrays in $\mathbf{f}_{\sf{HCB}}(r,c) \in \mathbb{C}^{L\times 1}$, we use the same weights $ \mathbf{f}_{\sf{HCB}}(r,c)$ for the $L$ DMA elements in all $D$ waveguides. We apply $D$ copies of $\mathbf{f}_{\sf{HCB}}(r,c)$ to form $\fps$ for the desired weights as $\fps(r,c)={\sf{vec}}\left( [\mathbf{f}_{\sf{HCB}}(r,c),\ldots,\mathbf{f}_{\sf{HCB}}(r,c)] \right)$. Applying $\mathbf{f}_{\sf{HCB}}(r,c)$ to all waveguides of the DMA creates larger beamforming gains in the desired direction of the DEACT beams through the DMA. We simulate beams along the azimuth direction through the DEACT approach, and leave beamforming in both the elevation and azimuth directions through a codebook design for future work. The entries in $\fps$ are the desired weights $\wbl$, which can be mapped through the Euclidean and Lorentzian-constrained modulation techniques to find $\qbl$. As $\qbl$ are the resulting entries to $\fdma$, we define the final Lorentzian-constrained and Euclidean modulated beamformers as $\fdma^{\sf{LC}}=-\frac{\sfj-e^{\sfj \angle(\fps \odot \mathbf{m}^c)}}{2}$, $\fdma^{\sf{EM}}= -\frac{\sfj+e^{\sfj ( \angle(2(\fps \odot \mathbf{m}^c)+\sfj)-\pi )}}{2}$. In the following section, we analyze the different mapping techniques and their effects on coverage, beamforming gain, and steering accuracy for the different codebook layers.

\section{Results}\label{sec: results}

We now analyze the hierarchical codebook design and mapping techniques with the realistic DMA. We simulate the Euclidean, Lorentzian-constrained and brute-force phase rotation optimization modulation techniques with the designed DMA in the electromagnetic simulation software HFSS to extract the resulting beam patterns. As typical directional base station antennas are designed to cover an angular region of around $120^\circ$ in the azimuth direction \cite{BeckmanLindmarkEvolutionBaseStationAntennas2007}, we analyze the DMA beam patterns primarily for their performance in the azimuth angle region $\phi \in [-60^\circ,60^\circ]$. We assume the DMA is oriented along the x-z plane such that the simulated beam patterns span the azimuth angular region. We also create a realistic phased array for direct comparison with the DMA through a simulated patch antenna array. We first design a patch antenna at the desired resonant frequency of 15 GHz in HFSS. We then create a $4 \times 5$ array with $\frac{\lambda}{2}$ spacing using this element as the realistic phased array. Since the DMA allows for element spacing smaller than half of a wavelength, the $4\times 5$ phased array maintains the same total aperture area as the $8 \times 5$ DMA. For the fourth hierarchical codebook layer, we simulate the desired beamforming angles with the realistic phased array to extract the beam patterns in the same way as the DMA. As the realistic phased array has a directional patch antenna element and the same aperture area as the DMA, this offers a fair comparison for coverage, spectral, and energy efficiency.

We briefly discuss the notation used in the following section. We have two different mapping techniques, Euclidean and Lorentzian-constrained modulation, and the additional mapping optimization with a phase rotation. We denote EM as Euclidean modulation and LC as Lorentzian-constrained modulation to represent the mapping technique used. We then denote NR to represent the case with no phase rotation and BF as the brute-force search optimization method for the phase rotation. As the brute-force method involves a computationally intensive search with electromagnetic simulation results, we only implement this optimization with Euclidean modulation. We further denote PS as the realistic patch antenna array with phase shifters, where A represents active phase shifters and P represents passive phase shifters. We also use HCB to indicate the hierarchical codebook. This notation is summarized in Table \ref{table: notation}.

\begin{table}[h!]
    \centering
    \caption{Notation used to represent simulation results.}
    \begin{tabular}{|c|c|}
        \hline
        HCB & Hierarchical codebook \\
        \hline
        EM: NR  & Euclidean mod.: No phase rotation \\
        \hline
        EM: BF  & Euclidean mod.: Brute-force search \\ \hline
        LC: NR  & Lorentzian-constrained mod.: No phase rotation \\ \hline

        PS: P  & Passive phased array with patch antenna elements \\ \hline
        PS: A  & Active phased array with patch antenna elements \\ \hline
        
    \end{tabular}
    \label{table: notation}
\end{table}

\subsection{DMA with wide beam patterns}

\begin{figure} 
    \centering
  \subfloat[\label{fig: lc layer 2}]{%
       \includegraphics[width=.48\linewidth]{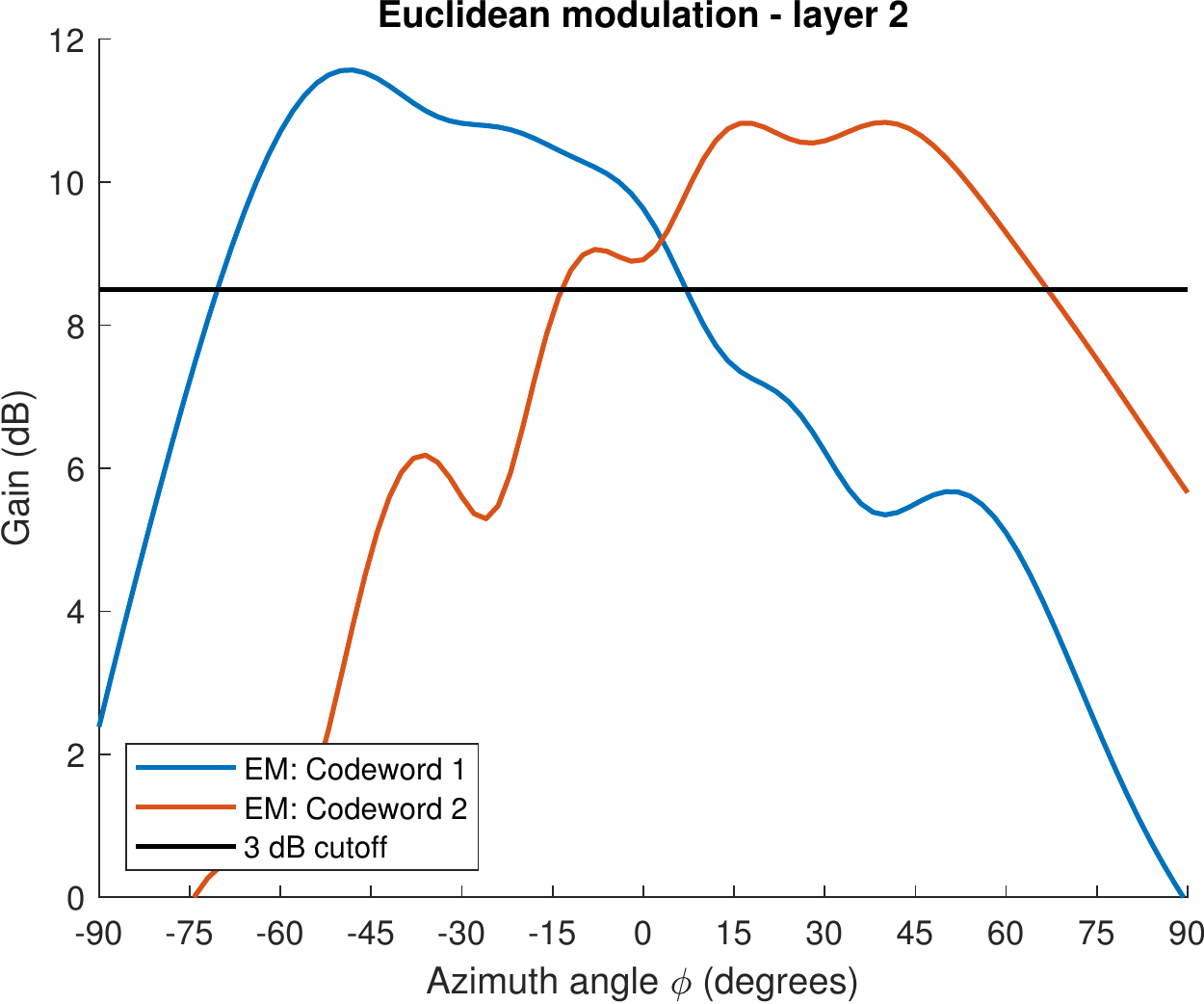}}
    \hfill
  \subfloat[\label{fig: em layer 2}]{%
        \includegraphics[width=.48\linewidth]{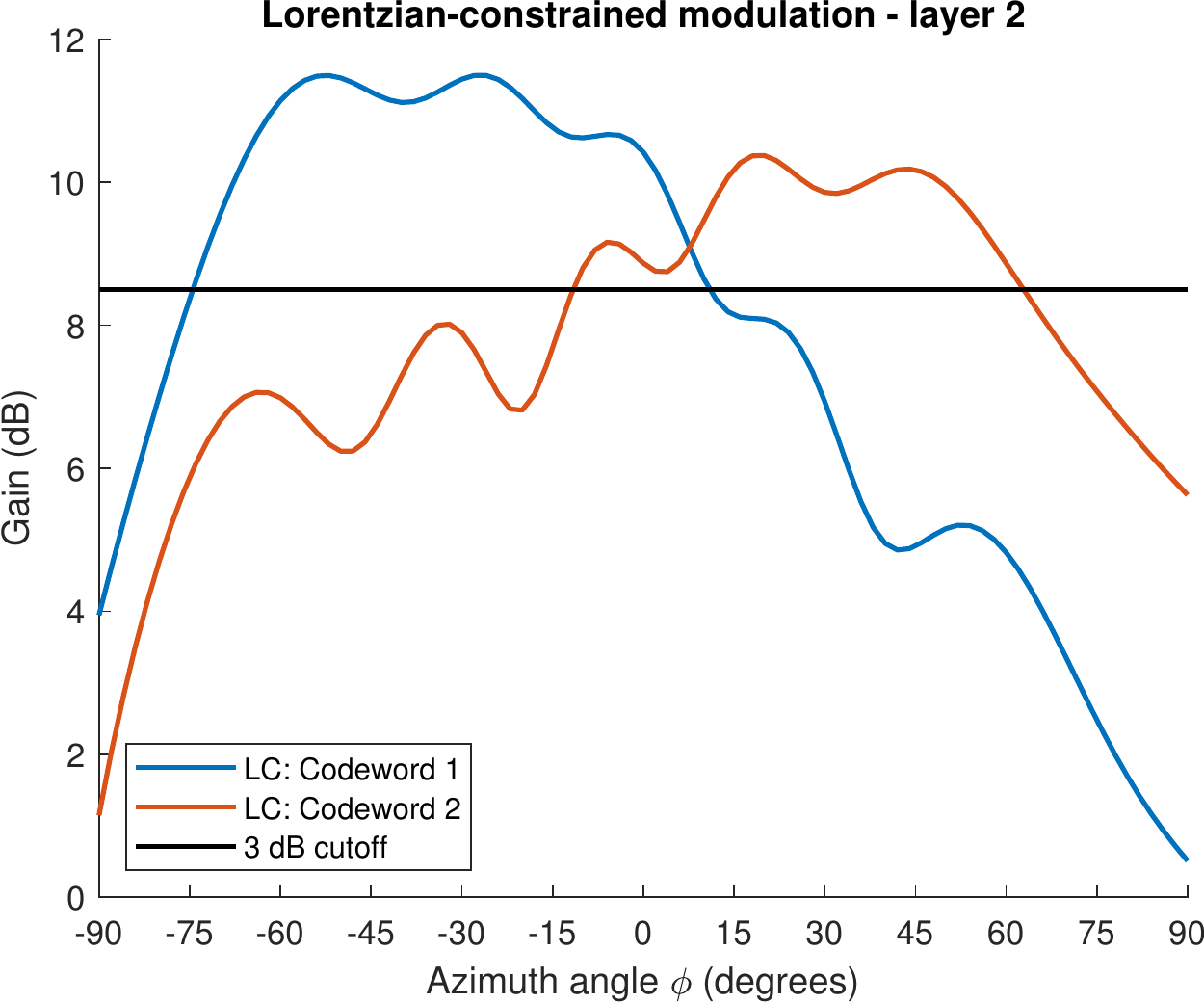}}

  \caption{Simulated beam patterns for the second hierarchical codebook layer using (a) Euclidean modulation and (b) Lorentzian-constrained modulation. We find that both mapping techniques enable large beamwidths for the DMA that adequately span a wide angular range.}
  \label{fig: layer 2} 
\end{figure}

We first analyze the ability of the DMA to provide wide beamwidths for the hierarchical codebook. We determine this through the half-power beamwidth of the DMA beams. For layer 2 of the hierarchical codebook, Fig. \ref{fig: layer 2} displays the HFSS simulated results for the two codeword beams using Euclidean and Lorentzian-constrained modulation. We present the normal mapping techniques here without any phase rotation optimization. We expect the two beams here to adequately span the angular region from $\phi \in [-60^\circ,60^\circ]$, and have beamwidths that intersect near the $\phi=0^\circ$ angle. As both the Euclidean and Lorentzian-constrained modulation beam patterns have a maximum of around $11.5$ dB, we plot a $3$ dB line as well for beamwidth analysis. We find in Fig. \ref{fig: layer 2} that Euclidean modulation provides beamwidths of $76^\circ$ and $80^\circ$ for codewords 1 and 2, which compares closely with Lorentzian-constrained modulation that has beamwidths of $84^\circ$ and $76^\circ$. Therefore, both mapping techniques enable large beamwidths for the two beams that remain above their $3$ dB gain for the desired angular range. We also find that the beam patterns intersect around the $\phi=0^\circ$ value, as desired. We demonstrate in these plots the capability of the DMA to create wide beamwidth beam patterns, which has not been discussed in prior work.

Although the beam patterns from the simulated results provide adequately wide beams for the hierarchical codebook, there is still some distortion in the beam pattern for each codeword. This distortion is mostly due to the weight gap near the zero-weight for the simulated weights in Fig. \ref{fig: dma weight sim}. When implementing the codebook layers, we chose the closest weight to zero since the zero-weight is unavailable. This led to small amounts of radiation from the ``off'' elements that distort the DMA beam. Fortunately, for our DMA design, the weight gap did not significantly degrade the performance of the second hierarchical codebook layer. This weight-gap, however, is not included in typical DMA models, so we highlight its ability to distort wider beamwidths through Fig. \ref{fig: layer 2}.

\subsection{DMA steering accuracy}

\begin{figure} 
    \centering
  \subfloat[\label{fig: em nr layer 4}]{%
       \includegraphics[width=.48\linewidth]{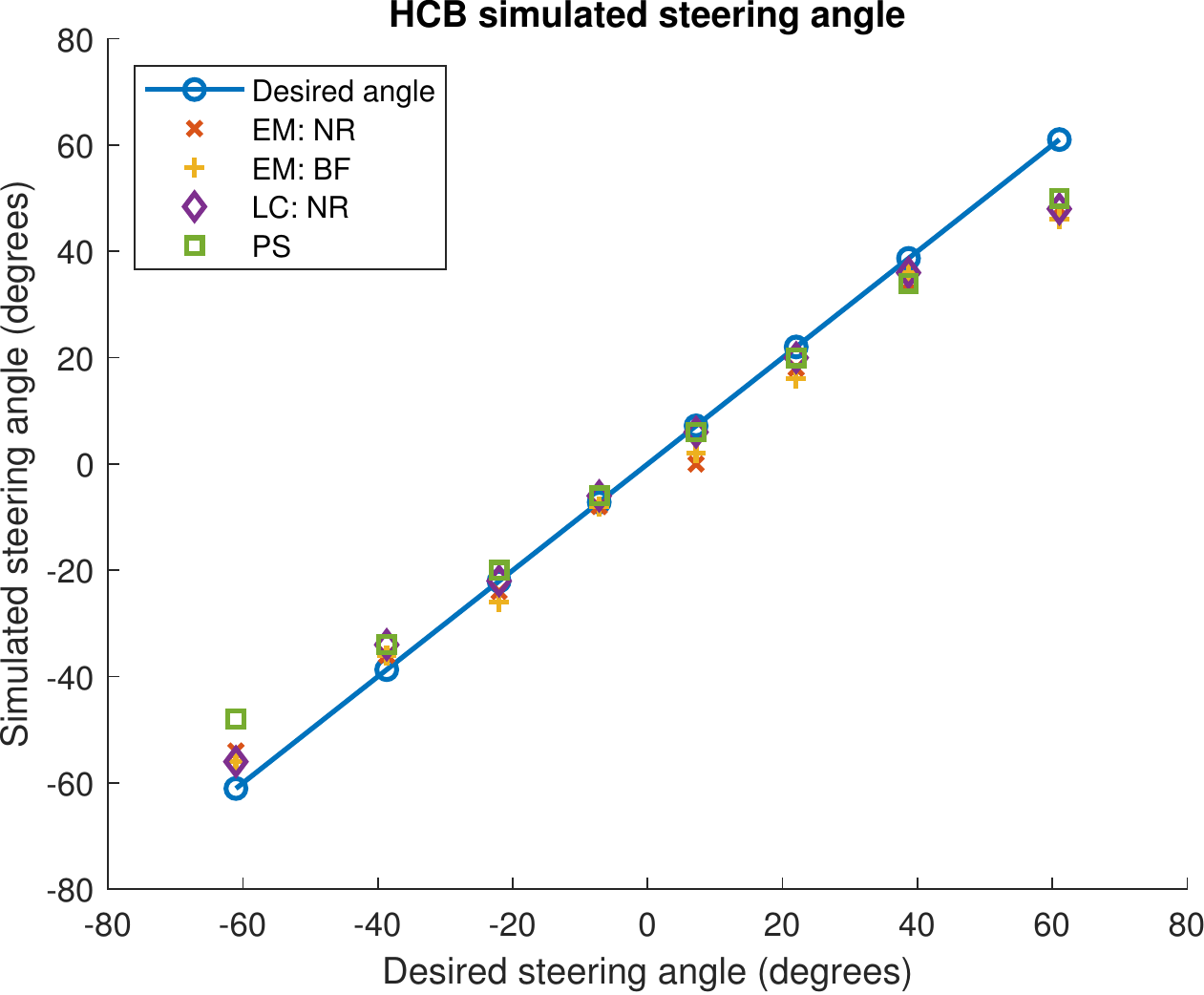}}
       \hfill
  \subfloat[\label{fig: em bf layer 4}]{%
       \includegraphics[width=.48\linewidth]{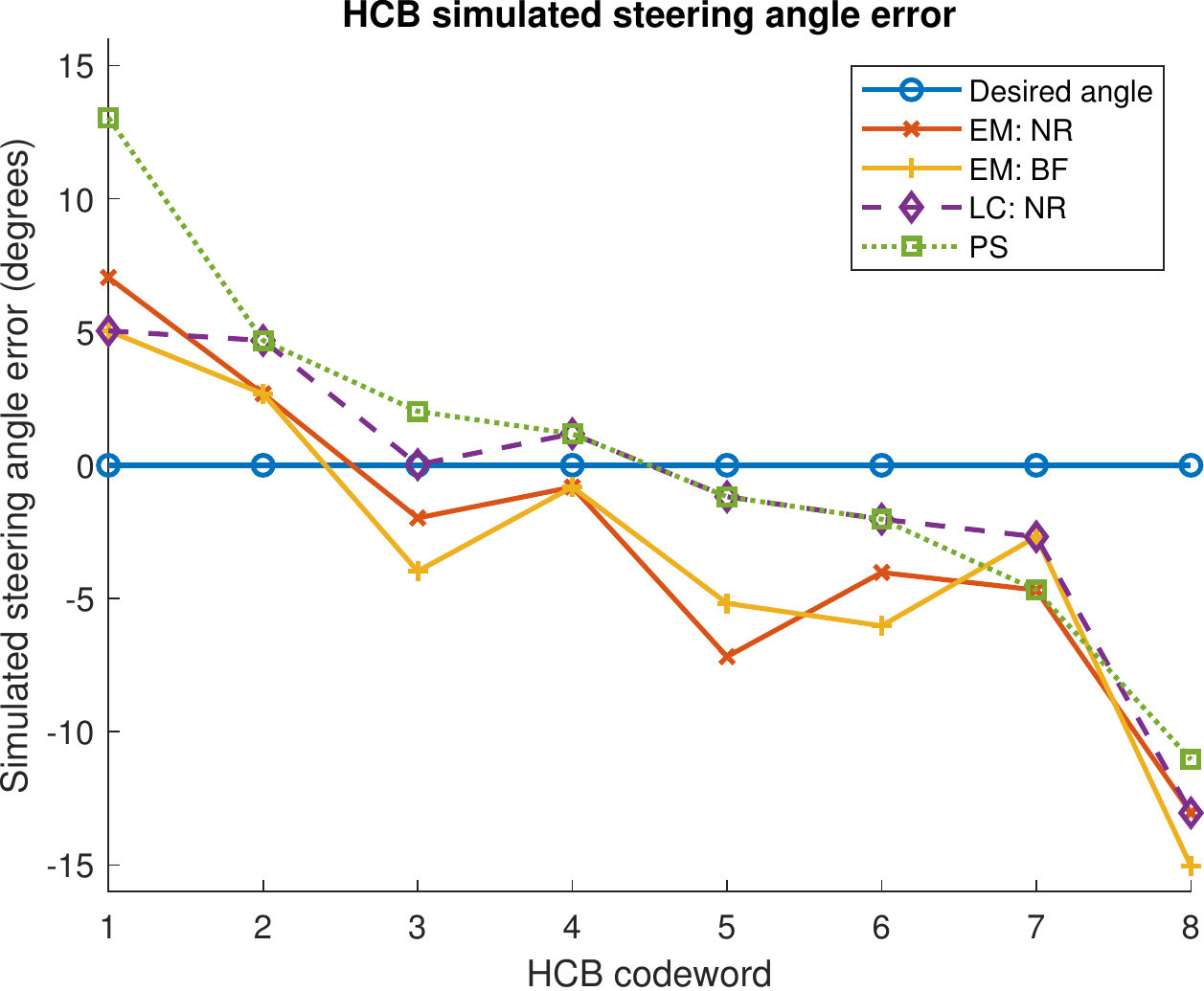}}

  \caption{(a) Desired versus simulated steering angle for the fourth hierarchical codebook layer and (b) steering angle error for each codeword. Lorentzian-constrained modulation generally provides the highest steering accuracy of the DMA mapping techniques. All mapping methods enable accurate beamsteering compared to the phased array.}
  \label{fig: steering angle} 
\end{figure}

We determine the steering accuracy of the DMA by looking at the fourth layer of the hierarchical codebook and the eight beams generated. Fig. \ref{fig: steering angle} displays the desired versus simulated steering angle for the various mappings and the resulting beamsteering error. Overall, we find that the DMA performs well with steering the beam in the correct direction despite the Lorentzian-constrained weights. Moreover, the DMA provides a similar steering accuracy to the patch antenna phased array. Therefore, the Lorentzian-constrained weights do not cause steering accuracy to diminish compared to typical antenna arrays. We note that both the DMA and patch elements have directional beam patterns, which naturally diminishes the beamsteering accuracy towards the end-fire angles. As expected from \cite{BowenEtAlOptimizingPolarizabilityDistributionsMetasurface2022}, Lorentzian-constrained modulation enables the highest steering accuracy as it provides a constant phase difference across the DMA elements.

\subsection{DMA coverage}

\begin{figure} 
    \centering
  \subfloat[\label{fig: hcb layer 2}]{%
       \includegraphics[width=.46\linewidth]{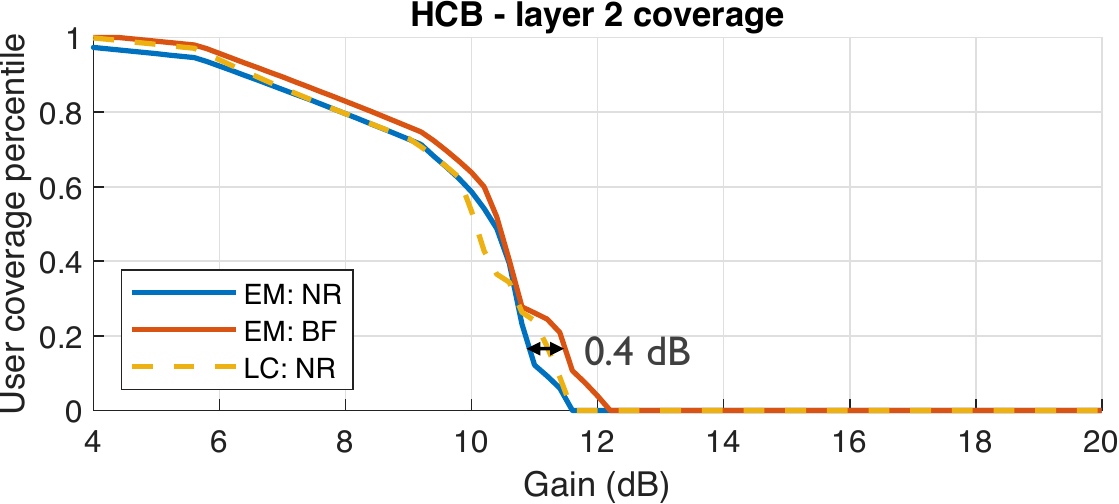}}
    \hfill
  \subfloat[\label{fig: hcb layer 3}]{%
        \includegraphics[width=.46\linewidth]{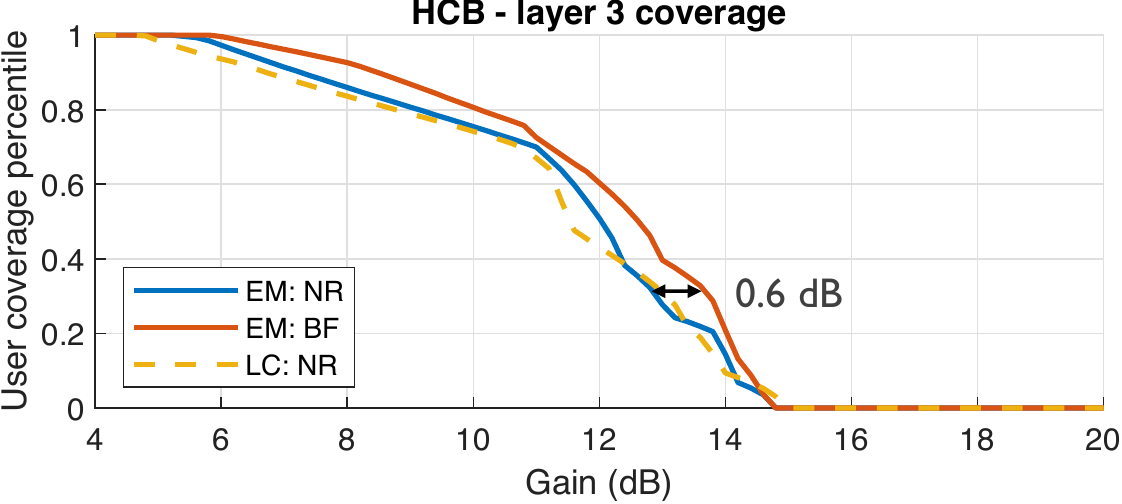}}

  \subfloat[\label{fig: hcb layer 4}]{%
        \includegraphics[width=.46\linewidth]{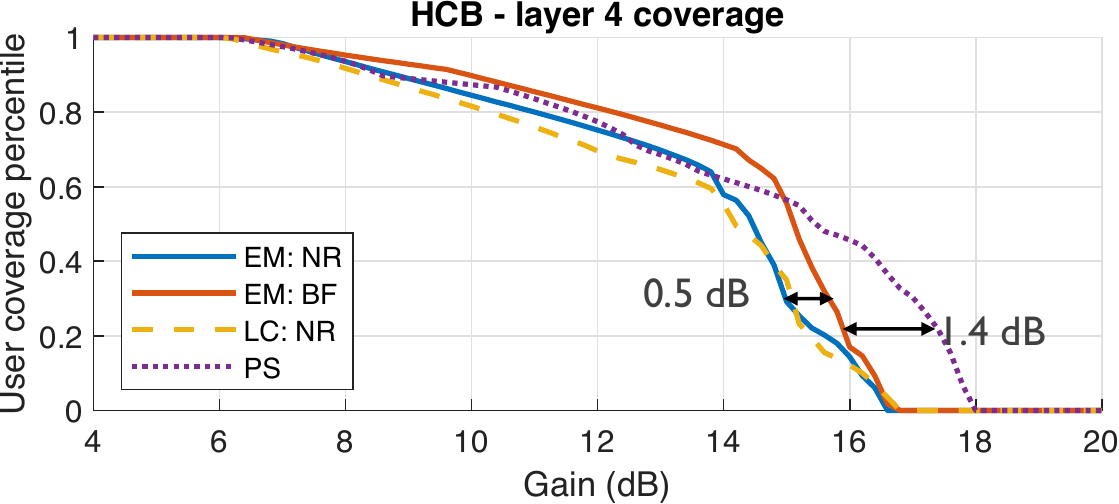}}

  \caption{Coverage analysis of simulated beam patterns for the hierarchical codebook layers. We find that the brute-force search provides the best coverage of the mapping techniques for the DMA.}
  \label{fig: hcb cdf} 
\end{figure}

Next, we analyze the coverage provided by the DMA hierarchical codebook for the different mapping techniques. To determine their effectiveness in spanning the desired angular range of $\phi \in [-60^\circ,60^\circ]$, we use a user coverage percentile curve as $1-\text{CDF}$ of the codebook gain CDF. The curve indicates the codebook coverage via the percentile of users in the desired angular range that could receive a certain beamforming gain. Fig. \ref{fig: hcb cdf} displays these results. In this figure, the curve furthest to the right generally provides the best coverage as a codebook and enables on average the highest beam pattern gain. For the three DMA mapping techniques, we find that the brute-force method provides the best coverage for all three hierarchical codebook layers, as it enables greater beamforming gain to cover the desired angular space. Since the brute-force search optimized the beam pattern through actual simulations rather than the theoretical DMA model, it  guarantees higher beamforming gain over the other mapping methods, leading to better coverage with the DMA beam patterns. We find that the phase rotation optimization can improve coverage compared to the usual mapping techniques when done with a brute-force search over the phase rotations. Though computationally intensive, the brute-force search simulations can be performed offline prior to the deployment of a DMA, and provides a faster method than checking all possible weight combinations to find the best DMA beam pattern.

We now analyze the fourth hierarchical codebook layer coverage with the phased array implementation. As expected, we find that the phased array ultimately provides higher beamforming gain of around 1.4 dB in many of the beam patterns compared to the DMA. This is due to the highly-efficient patch antenna design, where DMAs exhibit loss from the resonant response of each element. Despite this, since patch antennas are slightly more directional than the DMA, for lower gains in the 10-14 dB range, the DMA enables better coverage near the end-fire angles. Although the phased array may generally provide better coverage from increased beamforming gain, this comes at the cost of transmit power loss and increased power consumption. We further analyze the effects of power consumption and overall performance of the phased array and DMA through spectral and energy efficiency results.

\subsection{Spectral efficiency}

We analyze spectral and energy efficiency for the DMA and realistic phased array through the channel simulation software QuaDRiGa \cite{BurkhardtEtAlQuaDRiGaMIMOChannelModel2014}. QuaDRiGa is a sophisticated simulation tool that allows for realistic channel modeling with a variety of settings and parameters. We choose to use this tool as the channel generation model is based on real-world channel measurements and is a tool commonly used by many researchers due to its accuracy and reproducibility. We are also able to easily implement the simulated beam patterns from the designed DMA and phased array into QuaDRiGa to examine the codebook design. We use a 3GPP urban macrocell environment as the simulation channel generation setting. We set a radial distance of 500 meters for each user from the base station DMA. Additionally, we take the number of OFDM subcarriers as $K=64$ for a narrow channel bandwidth of $B = 20$ MHz and assume the simulated gain patterns do not change significantly with frequency over the small bandwidth.

For the simulations, we assume a limited feedback beamforming scheme where each beam pattern in the codebook is tested for the channel and the beam pattern with the highest channel gain is selected. The final output of the simulation is a post-beamforming channel value. We gather this data for the fourth hierarchical codebook layer for all simulated DMA codebooks and the realistic phased array codebook. We define the system SNR for a transmit power $P_{\sf{T}}$ and thermal noise power $N_0$ as $\rho = \frac{P_{\sf{T}}}{B N_0}$. It is important to note that the transmit power will be different for the DMA, active phased array, and passive phased array due to the varying component losses between these three cases. Therefore, each scenario will have different SNR values for calculating spectral efficiency from the simulated channel measurements using \eqref{eq: spec eff}. We compare the scenarios fairly by plotting the resulting codebook spectral efficiency values against the system input power $P_{\sf{IN}}$.

Using the resulting QuaDRiGa channel simulations, the input-output transmit power relation \eqref{eq: trans power}, and the channel spectral efficiency equation  \eqref{eq: spec eff}, we show the spectral efficiency values against the input power in Fig. \ref{fig: spec eff}. As expected from the coverage results, we find that the brute-force method for Euclidean modulation provides the highest spectral efficiency of the DMA codebooks. We also see that the best-performing codebook is the realistic phased array with active phase shifters, resulting in a 3.6 dB input power difference with the brute-force DMA method. This is anticipated since the active phase shifters give additional gain to the transmit power signal (2.3 dB), and the realistic phased array beam patterns generally provide greater beamforming gain than the DMA. Naturally, this case also consumes the most power, which we will analyze in the next section through energy efficiency. The worst-performing case is the realistic phased array with high-resolution passive phase shifters, allowing for a 7.4 dB input power difference with the brute-force DMA method. This is due to the significant component loss (8.8 dB) from the phase shifters, which attenuates the input power signal and gives lower SNR values. We can conclude from these results that the designed DMA enables higher spectral efficiency values than the traditional phased array equipped with high-resolution passive phase shifters. We also see that the proposed phase rotation optimization through brute-force simulation results provides around a 0.4 dB increase in spectral efficiency in terms of the input power. Although small, this optimization can be used to increase the achievable spectral efficiency of a DMA with no tradeoffs on any other system performance metric. 

\begin{figure} 
    \centering
  \subfloat[\label{fig: spec eff}]{%
       \includegraphics[width=.48\linewidth]{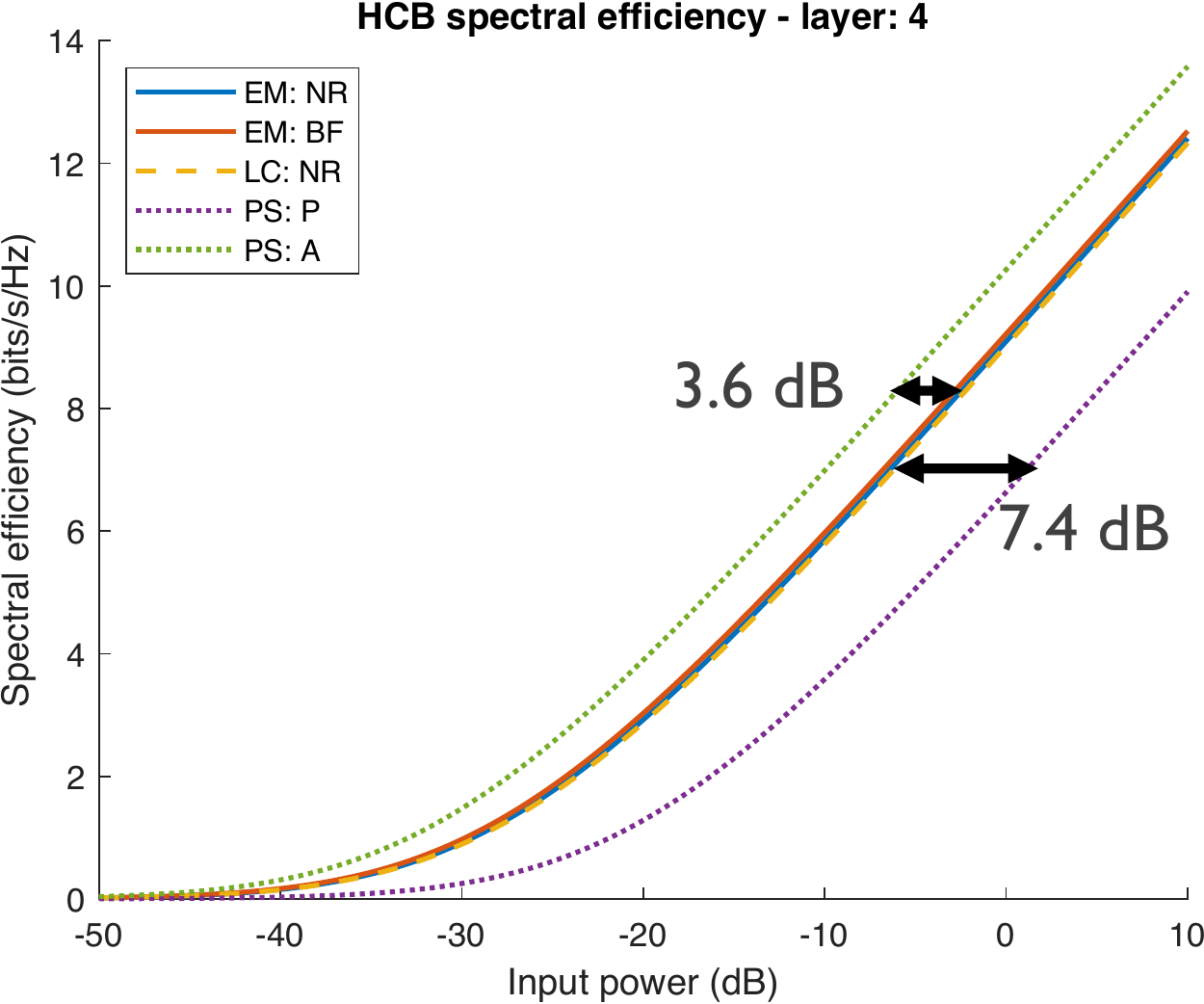}}
       \hfill
  \subfloat[\label{fig: en eff}]{%
       \includegraphics[width=.48\linewidth]{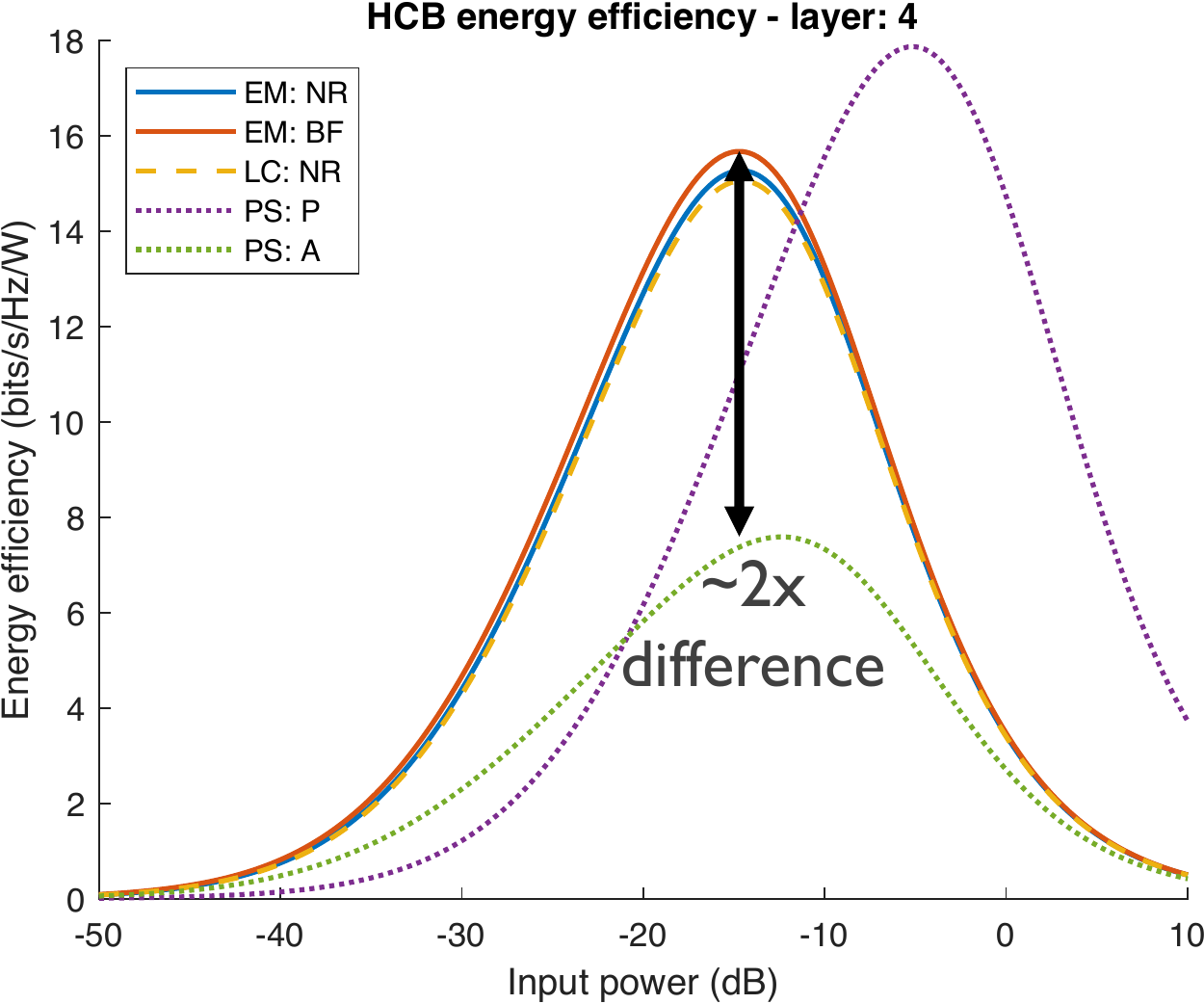}}

  \caption{Simulated results with the DMA and phased array for (a) spectral efficiency and (b) energy efficiency. We find that the DMA outperforms a traditional passive phased array due to the significant transmit power attenuation for the high-resolution phase shifter. While the active phased array provides the highest spectral efficiency, its large power consumption makes it the least energy-efficient option.}
  \label{fig: eff} 
\end{figure}

\subsection{Energy efficiency}

Lastly, we analyze the energy efficiency of the DMA and realistic phased array codebooks using \eqref{eq: EE msa} and \eqref{eq: EE sba}. Fig. \ref{fig: en eff} shows the energy efficiency results from the spectral efficiency and power consumption calculations. While the realistic phased array with active phase shifters yielded the highest spectral efficiency, we see here that the increased power consumption for this case causes the lowest energy efficiency values across the input power. Therefore, the significant power consumption savings due to the reconfigurable component and fewer power dividers causes the DMA to achieve around twice the energy efficiency of the active phased array at the peak energy efficiency value. While the passive phased array curve actually provides better energy efficiency than the DMA at higher input power values, this is mostly due to component loss. As the passive phase shifters attenuate the input power by 8.8 dB, the power consumption due to power amplifiers is significantly reduced for this case, leading to an overall lower power consumption than the DMA. The passive phased array case, however, still provides lower spectral efficiency than the DMA by around 7.4 dB in terms of input power. Thus, its energy efficiency does not make up for this significant gap in spectral efficiency performance. The spectral and energy efficiency results in Fig. \ref{fig: spec eff} and \ref{fig: en eff} show the capabilities of the DMA to form energy-efficient arrays that maintain or surpass the performance of typical phased arrays.

\section{Conclusions}\label{sec: conclusion}

In this paper, we investigated the use of DMAs with a hierarchical codebook and proposed a novel technique for enhancing the beamforming gain compared to current methods. Through the implementation of simulated DMA beam patterns and channel environments, we found that the proposed mapping technique of a brute-force phase rotation search provides better coverage and spectral efficiency than the unoptimized Euclidean and Lorentzian-constrained modulation. Moreover, in terms of spectral efficiency, we found that the DMA significantly outperforms a traditional phased array with passive phase shifters due to high component loss. While providing lower spectral efficiency than an active phased array, the DMA results in greater energy efficiency due to the increased power consumption from active phase shifters. Therefore, DMAs are shown as a suitable replacement for phased arrays in situations where power consumption should be minimal. In future work, we plan to develop novel precoding techniques further optimized around realistic DMA effects. This includes aspects like element perturbation of the waveguide, DMA weight tuning constraints, and refining the DMA geometry. Furthermore, we aim to revisit many MIMO applications, such as channel estimation or localization, with the DMA weight constraint to analyze potential benefits and tradeoffs for using DMAs in these scenarios.

\section*{Appendix}

\subsection{Proof of Lemma 1}
For the Euclidean modulation equation described in Section \ref{sec: DMA mapping} as $q^{\sf{EM}} = \argmin\limits_{q \in \mQ}|q-\tilde{w}|^2$, we substitute and solve for the DMA weight angle $\psi$ as

\begin{equation}
    \psi^{\sf{EM}} = \argmin\limits_{\psi \in [0,2\pi)}|-\frac{\sfj+e^{\sfj\psi}}{2}-\tilde{w}|^2.
\end{equation}

\noindent We further define $\bar{w}=2 \tilde{w}+\sfj$, so that ${\psi}^{\sf{EM}}$ becomes

\begin{equation}
\begin{split}
    \psi^{\sf{EM}} & = \argmin\limits_{\psi \in [0,2\pi)}|e^{\sfj \psi}-\bar{w}|^2 \\
    & = \argmin\limits_{\psi \in [0,2\pi)}\operatorname{Re}\left\{e^{-\sfj \psi}\bar{w} \right\}.
\end{split}
\end{equation}

\noindent Then, for $\bar{w} = |\bar{w}|e^{\sfj \angle{\bar{w}}}$, we have $\operatorname{Re}\left\{e^{-\sfj \psi}\bar{w} \right\} = |\bar{w}| \operatorname{Re}\left\{e^{ \sfj \left(\angle{\bar{w}} - \psi \right)} \right\}=|\bar{w}|\cos\left(\angle{\bar{w}} - \psi \right)$. The maximum value of 1 for $|\cos\left(\angle{\bar{w}} - \psi \right)|$ is achieved for the argument $\left(\angle{\bar{w}} - \psi \right)=\pi+2\pi k, k\in \mathbb{Z}$. Therefore, we set $\psi=\angle{\bar{w}}-\pi$ to find the Euclidean modulation weight. From the substitutions for $\psi$ and $\bar{w}$, the final DMA weight is then $q^{\sf{EM}} = -\frac{\sfj+e^{ \sfj ( \angle (2 \tilde{w}+\sfj )-\pi ) }}{2}$.

\bibliographystyle{IEEEtran}

\bibliography{refs}

\end{document}